\documentclass[conference]{IEEEtran}
\IEEEoverridecommandlockouts
\usepackage{cite}
\usepackage{amsmath,amssymb,amsfonts}
\usepackage[colorlinks=true,allcolors=blue]{hyperref}
\usepackage{amsthm}
\usepackage{bm}
\usepackage{bbm}
\usepackage{comment}
\usepackage{booktabs}
\usepackage[linesnumbered]{algorithm2e}
\usepackage{graphicx}
\usepackage{textcomp}
\usepackage{xcolor}
\newtheorem{theorem}{Theorem}

\newtheorem{proposition}{Proposition}
\newtheorem{lemma}{Lemma}
\newtheorem{definition}{Definition}
\def\BibTeX{{\rm B\kern-.05em{\sc i\kern-.025em b}\kern-.08em
    T\kern-.1667em\lower.7ex\hbox{E}\kern-.125emX}}
\begin{document}

\title{Power Allocation for the Base Matrix of Spatially Coupled Sparse Regression Codes
}

\author{Nian Guo, Shansuo Liang, Wei Han\thanks{Nian Guo, Shansuo Liang, and Wei Han are with Theory Lab, Central Research Institute, 2012 Labs, Huawei Technologies Co. LTD., Hong Kong SAR, China. E-mail: \{guonian4, liang.shansuo, harvey.huawei\}@huawei.com. }
}

\maketitle

\begin{abstract}
We investigate power allocation for the base matrix of a spatially coupled sparse regression code (SC-SPARC) for reliable communications over an additive white Gaussian noise channel. A conventional SC-SPARC allocates power uniformly to the non-zero entries of its base matrix. Yet, to achieve the channel capacity with uniform power allocation, the coupling width and the coupling length of the base matrix must satisfy regularity conditions and tend to infinity as the rate approaches the capacity. For a base matrix with a pair of finite and arbitrarily chosen coupling width and coupling length, we propose a novel power allocation policy, termed V-power allocation. V-power allocation puts more power to the outer columns of the base matrix to jumpstart the decoding process and less power to the inner columns, resembling the shape of the letter V. We show that V-power allocation outperforms uniform power allocation since it ensures successful decoding for a wider range of signal-to-noise ratios given a code rate in the limit of large blocklength. In the finite blocklength regime, we show by simulations that power allocations imitating the shape of the letter V improve the error performance of a SC-SPARC.

\end{abstract}

%\begin{IEEEkeywords}
%Sparse regression code, sparse superposition code, spatial coupling, base matrix, power allocation, approximated message passing algorithm, state evolution.
%\end{IEEEkeywords}

\section{Introduction}
For reliable communications over an additive white Gaussian noise (AWGN) channel, Joseph and Barron \cite{Joseph12} designed the sparse regression code (SPARC). It forms a codeword by multiplying a design matrix by a sparse message. The message is sparse as it is segmented into several sections and each section contains only one non-zero entry. The codeword is passed through an AWGN channel subject to an average power constraint. With uniform power allocation across the non-zero entries of a message and a maximum likelihood decoder, a SPARC asymptotically achieves the channel capacity of the AWGN channel~\cite{Joseph12}. To overcome the complexity barrier of the maximum likelihood decoder, the approximate message passing (AMP) decoder with polynomial complexity has been proposed \cite{Barbier14}--\hspace{1sp}\cite{Barbier17}. Its decoding error is closely tracked by the state evolution (SE) and it outperforms other low-complexity decoders \cite{Joseph14}\hspace{1sp}\cite{Cho13} in terms of the finite-blocklength error rates. By judiciously allocating power to the non-zero entries of a sparse message, SPARCs with AMP decoding continue to achieve the channel capacity \cite{Rush17}. For example, iterative power allocation \cite{Ramji19} uses the asymptotic SE of the AMP decoder to decide the power allocation for a message section by section.

By introducing a spatial coupling structure to the design matrix, SC-SPARCs with AMP decoding not only achieve the channel capacity \cite{Barbier16}\hspace{1sp}\cite{Rush21} but also display a better error performance compared to power-allocated SPARCs \cite{Barbier17}\hspace{1sp}\cite{Hsieh18}. Similar to the graph-lifting of SC-LDPC codes \cite{Kudekar}\hspace{1sp}\cite{Mitchell}, the design matrix of a SC-SPARC is constructed from a base matrix. Each entry of the base matrix is expanded as a Gaussian submatrix in the design matrix, and the variance of the Gaussian entry is determined by the corresponding entry in the base matrix. The coupling structure of the base matrix is determined by a coupling pair comprising a coupling width and a coupling length.

Existing works on SC-SPARCs commonly assumed that the power is uniformly allocated to the non-zero entries of the message as well as the base matrix, e.g., \cite{Ramji19}\hspace{1sp}\cite{Rush21}\hspace{1sp}\cite{Hsieh18}\hspace{1sp}\cite{Rush19}. For such uniform power allocation (UPA), a decoding phenomenon termed \emph{sliding window} is observed \cite{Rush21}\hspace{1sp}\cite{Rush19}, namely, the decoding propagates from two sides to the middle of a message in a symmetric fashion. Once the outer parts of a message are successfully decoded, they act as perfect side information that facilitates the decoding of the inner parts of the message. This phenomenon is used as a decoding techinque termed \emph{seed} to boost the decoding performance of SC-SPARCs \cite{Barbier16}.

While UPA is sufficient for a SC-SPARC with AMP decoding to achieve the channel capacity, the coupling pair of the base matrix must satisfy regularity conditions and tend to infinity as the rate approaches the channel capacity \cite{Barbier16}\hspace{1sp}\cite{Rush21}. Yet, in practical implementations, the coupling pair is finite and arbitrary. Given a finite coupling pair, it has been observed that UPA might be inefficient and causes AMP decoding failure. Thus, it is of practical interest to design a power allocation policy for a base matrix with a finite coupling pair to ensure successful decoding for a wide range of power and code rates.

We propose a novel power allocation policy--V-power allocation (VPA)--for the base matrix of a SC-SPARC with AMP decoding. Its power allocation is non-increasing from the outer columns to the middle column of the base matrix, resembling the shape of the letter V. Similar to iterative power allocation \cite{Ramji19}, VPA leverages the asymptotic SE of the AMP decoder to tell whether a SC-SPARC ensures successful decoding in the limit of large blocklength. Dissimilar to conventional power allocation policies that vary the non-zero coefficients of a message, VPA only varies the non-zero entries of a base matrix. To measure the performance of a power allocation policy for the base matrix, we define a power-rate function (PRF). Given a finite coupling pair, a channel noise variance, and a rate, the PRF quantifies the minimum power so that a SC-SPARC with a power allocation policy ensures successful decoding for all power above it. We derive the PRFs for UPA and VPA, respectively, and we show that VPA outperforms UPA in terms of the PRF, meaning that VPA ensures successful decoding for a larger range of power. While VPA is designed in the infinite blocklength regime, we use simulations to show that a VPA-like power allocation improves the finite-blocklength block error rates of a SC-SPARC.

%We propose a novel power allocation algorithm, termed \emph{V-power allocation} (VPA), for the base matrix of a SC-SPARC with AMP decoding whose coupling pair is finite. It is called VPA since the power allocation is non-increasing from the outer columns to the middle column of the base matrix, similar to the shape of the letter V. We show that VPA outperforms UPA: given a rate, VPA ensures sucessful decoding for a wider range of SNRs. To this end, in Section~\ref{Sec_sparc}, we review the basics of SC-SPARCs. In Section~\ref{Sec_SE}, given rate $R$, we use the asymptotic SE of the AMP decoder to derive $\text{SNR}_U(R)$, so that for any $\text{SNR}>\text{SNR}_U(R)$, UPA leads to successful decoding in the limit of large blocklength. We also present the dual problem $R_U(\text{SNR})$. Since $R_U(\text{SNR})$ is smaller than the channel capacity, UPA no longer achieves the channel capacity for a finite coupling pair. In Section~\ref{Sec_alg}, we derive $\text{SNR}_V(R)$, the smallest SNR for VPA to ensure successful decoding for rate $R$. We show $\text{SNR}_V(R)>\text{SNR}_U(R)$ to conclude the advantage of VPA over UPA. While VPA is designed in the limit of large blocklength, in Section~\ref{Sec_sim}, we use simulations to show that power allocation imitating the shape of the letter V improves the finite-blocklength error performance of a SC-SPARC with AMP decoding.    

\emph{Notations:} For a positive integer $n$, we denote $[n]\triangleq \{1,2,\dots,n\}$. For a matrix $\mathsf W$, we denote by $\mathsf W_{rc}$ the entry at the $r$-th row and the $c$-th column. For a sequence $a_1,a_2,\dots$, we denote $\{a_i\}_{i=p}^{q}\triangleq \{a_p, a_{p+1},\dots,a_q\}$.

\section{Spatially coupled sparse regression codes}\label{Sec_sparc}
\subsection{Encoder}\label{Enc}
The encoder of a SC-SPARC forms a codeword $\bm x\in\mathbb R^n$ by multiplying a message vector $\bm \beta\in\mathbb R^{ML}$ by a design matrix $\mathsf A\in\mathbb R^{n\times ML}$,
\begin{align}\label{encoder}
\bm x = \mathsf A\bm \beta,
\end{align}
and the codeword is subject to an average power constraint
\begin{align}\label{X_power_constraint}
\frac{1}{n}\mathbb E[||\bm x||^2] = P.
\end{align} 

The message $\bm \beta$ is a sparse vector of length $ML$. It consists of $L$ length-$M$ sections. In each section $\ell=1,2,\dots, L$, there is only one non-zero entry, whose value is set a priori. Since the information is carried only by the indices of the non-zero entries, the alphabet size of $\bm \beta$ is $M^L$. As we will vary the variances of the entries of design matrix $\mathsf A$ by varying the power allocation for the base matrix, we set   all the non-zero coefficients of $\bm \beta$ to $1$ without loss of generality. 

The design matrix $\mathsf A$, as shown in Fig.~\ref{Fig_encoder}, is constructed from a base matrix $\mathsf W$.
The base matrix serves as a protograph for the design matrix. Each entry $\mathsf W_{rc}$ of base matrix $\mathsf W$ is expanded as an $M_R\times M_C$ submatrix of design matrix $\mathsf A$, whose entries are i.i.d. Gaussian random variables $\mathcal N\left(0,\frac{1}{L}\mathsf W_{rc}\right)$. A column block in $\mathsf A$ corresponds to a set of $M_C$ columns that are expanded from one column in $\mathsf W$. A row block in $\mathsf A$ corresponds to a set of $M_R$ rows that are expanded from one row in $\mathsf W$. The design matrix $\mathsf A$ contains $L_C$ columns blocks and $L_R$ row blocks. It holds that $M_CL_C = ML$, $n = M_RL_R$. % $M$ divides $M_C$, so that the part of message $\bm\beta_c$ corresponding to column block $c\in[L_C]$ contains a discrete number of sections. 

The rate of a SC-SPARC is defined as
\begin{align}
R = \frac{L\log M}{n}~\text{(nats per channel use)}.
\end{align}

In this work, we focus on a class of band-diagonal base matrices defined below, which is introduced in \cite{Rush21}. We denote by $\omega$ and $\Lambda$ the coupling width and the coupling length of the base matrix, respectively.   
\begin{definition}\label{def_base_matrix}
An $\left(\omega, \Lambda, P\right)$ base matrix $\mathsf W$ is specified by the following properties.
\begin{itemize}
\item[i)] The base matrix $\mathsf W$ is of size  $L_R\times L_C$, where $L_C \triangleq \Lambda$, $L_R \triangleq \omega + \Lambda - 1$, $\Lambda\geq 2\omega -1$;
\item[ii)]  Given any column $c\in[L_C]$, the non-zero entries are only at rows $c \leq r\leq c+\omega - 1$;
\item[iii)] The entries of $\mathsf W$ satisfy the average power constraint~\eqref{X_power_constraint},
\begin{align}\label{power_constraint}
\frac{1}{L_RL_C}\sum_{r=1}^{L_R}\sum_{c=1}^{L_C}\mathsf W_{rc} = P.
\end{align}
\end{itemize}
\end{definition}

\begin{figure}
\includegraphics[scale=0.45]{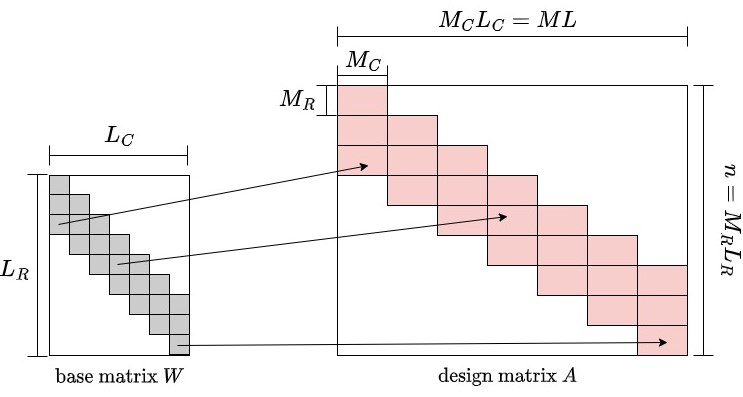}
\caption{Base matrix and design matrix of a SC-SPARC. The base matrix has coupling width $\omega=3$ and coupling length $\Lambda = 7$.}
\label{Fig_encoder}
\end{figure}
\subsection{Decoder}\label{Dec}
The codeword $\bm x$ \eqref{encoder} is transmitted through an AWGN channel yielding $\bm y = \bm x + \bm w$,
where $\bm w$ is a vector of $n$ i.i.d. Gaussian random variables each with zero mean and variance $\sigma^2$. 
The AMP decoder iteratively estimates the message $\bm \beta$ from the channel output $\bm y\in\mathbb R^n$ as follows \cite[Section III]{Rush21}. At iteration $t=0$, the AMP decoder initializes the estiamte of $\bm\beta$ as $\bm\beta^0 = \bm 0$ and initilizes two vectors $\bm v^0 =\bm 0$, $\bm z^{-1} = \bm 0$. At iterations $t=1,2,\dots$, the AMP decoder calculates the estimate $\bm \beta^t$ as
\begin{align}
&\bm z^{t} = \bm y - \mathsf A\bm \beta^{t} + \bm v^{t} \otimes \bm z^{t-1},\\
&\bm \beta^{t+1} = \eta_{t}\left(\bm \beta^{t} + (\mathsf S^{t}\otimes \mathsf A)^{*}\bm z^{t}\right), 
\end{align}
where $\otimes$ denotes the entry-wise product; function $\eta_{t}$ is the minimum mean square error estimator for $\bm \beta$;
vector $\bm v^t$ and matrix $\mathsf S^t$ are determined by the SE parameters. In the asymptotic regime $M\rightarrow\infty$, the SE parameters  \cite[(23)--(24)]{Rush21} at iterations $t=0,1,\dots$ are given by  
\begin{subequations}\label{asymptotic_SE}
\begin{align}\label{phi_SE}
&\phi_r^t  = \sigma^2 + \frac{1}{L_C}\sum_{c=1}^{L_C}\mathsf W_{rc}\psi_c^t, ~\forall r\in[L_R],\\ \label{psi_SE}
&\psi_c^{t+1}  = 1 - \mathbbm{1}\left\{\frac{1}{RL_R}\sum_{r=1}^{L_R}\frac{\mathsf W_{rc}}{\phi_r^t}>2\right\},~\forall c\in[L_C],
\end{align} 
\end{subequations}
where $\psi_c^0 = 1,~\forall c\in[L_c]$.
The SE parameter $\psi_c^t$ \eqref{psi_SE} closely tracks the normalized mean-square error between the part of message $\bm\beta_c$ and the part of the estimate $\bm \beta_c^t$ corresponding to column block $c$ at iteration $t$, i.e., $\psi_c^t\approx \frac{L_C}{L}||\bm \beta_c - \bm \beta_c^t||^2_2$ for all $c\in[L_C]$. This is evidenced both by the simulations \cite[Fig. 3]{Rush21} and the concentration inequality \cite[Theorem~2]{Rush21}.

\section{Power allocation and performance metrics}\label{Sec_metric}
We define power allocation policies for a base matrix as well as the performance metrics.

For an $(\omega,\Lambda,P)$ base matrix $\mathsf W$ in Definition~\ref{def_base_matrix}, a \emph{power allocation policy} is a mapping $\Pi\colon \mathbb R\rightarrow\mathbb R^{L_R\times L_C}$ that gives a set of non-negaitve values $\Pi(P) = \{\mathsf W_{rc}\}_{r\in[L_R], c\in[L_C]}$ corresponding to the entries of the base matrix. The power allocation policy $\Pi$ for the base matrix does not affect the non-zero coefficients of message $\bm\beta$.

We say that a SC-SPARC \emph{successfully} decodes column block $c$ of the message, i.e., $\bm\beta_c$, if there exists a time $T\in\mathbb Z_+$ such that $\psi_c^T = 0$ \eqref{psi_SE}; we say that a SC-SPARC \emph{successfully} decodes the entire message if there exists a time $T\in\mathbb Z_+$, 
\begin{align}\label{success}
\psi_c^T = 0,~ \forall c\in[L_C].
\end{align}

We use the asymptotic SE parameter $\psi_c^t$ \eqref{psi_SE} to define the performance metrics. The asymptotic SE parameter $\psi_c^t$ is fully determined by the coupling pair $(\omega,\Lambda)$, the noise variance $\sigma^2$, the rate $R$, the power $P$, and the power allocation policy $\Pi$. Fixing the first three parameters, it becomes $\psi_c^t=\psi_c^t(R, P, \Pi)$.

We measure the performance of a power allocation policy using the rate-power function (RPF) and the power-rate function (PRF) defined next. 
\begin{definition}
Fix a finite coupling pair $(\omega,\Lambda)$, a noise variance of the AWGN channel $\sigma^2$, and a power $P$. The RPF $R_{\Pi}(P)$ for power allocation policy $\Pi$ is the largest rate so that for any rate $R<R_{\Pi}(P)$, a SC-SPARC generated by an $(\omega, \Lambda, P)$ base matrix with power allocation $\Pi$ ensures successful decoding, 
\begin{align}\nonumber
R_{\Pi}(P)\triangleq \label{RPF}
\sup\{R^*\colon &\forall R<R^*, \exists T\in\mathbb Z_+,\\
 &\psi_c^T(R,P,\Pi) = 0, \forall c\in[L_C]\}.
\end{align}
Fix a finite coupling pair $(\omega,\Lambda)$, a noise variance of the AWGN channel $\sigma^2$, and a rate $R$. The PRF $P_{\Pi}(R)$ for power allocation policy $\Pi$ is the minimum power so that for any power $P>P_{\Pi}(R)$, a SC-SPARC generated by an $(\omega, \Lambda, P)$ base matrix with power allocation $\Pi$ ensures successful decoding, 
\begin{align}\nonumber
P_{\Pi}(R) \triangleq \inf\{P^*\colon &\forall P>P^*, \exists T\in\mathbb Z_+,\\\label{PRF}
&\psi_c^T(R, P, \Pi)=0, \forall c\in[L_C]\}.
\end{align}  
\end{definition}

We aim to find a power allocation policy $\Pi$ that leads to a large $R_{\Pi}(R)$, or equivalently, a small $P_{\Pi}(R)$. 

\section{Uniform power allocation}\label{Sec_SE}
We say that an $(\omega, \Lambda, P)$ base matrix in Definition~\ref{def_base_matrix} has uniform power allocation (UPA) if
\begin{align}\label{uniform_power}
\mathsf W_{rc} = 
\begin{cases}
P\frac{L_R}{\omega}, & c\leq r\leq c+\omega-1,\\
0,&\text{otherwise}.
\end{cases}
\end{align}
We show the RPF \eqref{RPF} and the PRF \eqref{PRF} for UPA.
\begin{theorem}\label{thm_UPA}
Fix a finite coupling pair $(\omega,\Lambda)$ and an AWGN channel with noise variance $\sigma^2$. The RPF $R_U(P)$ for UPA is given by 
\begin{align}\label{RUPF}
R_U(P) = \frac{L_C}{2L_R}\sum_{r=1}^{\omega}\frac{1}{r+\frac{L_C}{L_R}\frac{\sigma^2}{P}\omega};
\end{align}
the PRF $P_U(R)$ for UPA is given by
\begin{align}\label{PURF}
P_U(R) = 
\begin{cases}
R_U^{-1}(R), & R<\frac{L_C}{2L_R}\sum_{r=1}^{\omega}\frac{1}{r},\\
\infty, & \text{otherwise}, 
\end{cases}
\end{align}
where $R_U^{-1}$ is the inverse function of $R_U$.
\end{theorem}
\begin{proof}
Appendix~\ref{pf_thm_UPA}.
\end{proof}
We compare $R_U(P)$ \eqref{RUPF} with the channel capacity $C(P) = \frac{1}{2}\log\left(1+\frac{P}{\sigma^2}\right)$ of the AWGN channel with noise variance $\sigma^2$. Using Right-endpoint approximation, we upper bound \eqref{RUPF} as
\begin{align}\label{Rb}
R_U(P) \leq \frac{L_C}{2L_R}\log\left(1+\frac{P}{\sigma^2}\frac{\omega}{\frac{L_C}{L_R}\omega + \frac{P}{\sigma^2}}\right).
\end{align}
The right side of \eqref{Rb} is smaller than $C(P)$ for a finite coupling pair, implying that a SC-SPARC with a finite coupling pair no longer achieves the channel capacity. The gap closes if and only if $\omega, \Lambda\rightarrow\infty$ and $\frac{\omega}{\Lambda}\rightarrow 0$.

For rates $R_U(P)\leq R<C(P)$, a SC-SPARC fails to ensure successful decoding, and the reason is shown in Proposition~\ref{prop2} stated below. We denote the index of the middle column of the base matrix by $\theta \triangleq \left\lceil\frac{\Lambda}{2}\right\rceil$.
\begin{proposition}\label{prop2}
Consider a SC-SPARC generated by an $\left(\omega, \Lambda, P\right)$ base matrix with UPA \eqref{uniform_power}. At iteration $t=1$, if the AMP decoder successfully decodes $2g$ column blocks of the message,
\begin{align}\label{prop2_psic1}
\psi_c^1 =\psi_{\Lambda-c+1}^t = 0,~\forall c\leq g,
\end{align}
for some $0\leq g \leq \omega$,
then at iterations $t=2,3,\dots$, the AMP decoder continues to decode $2g$ column blocks of the message,
\begin{align}\label{psi_gct}
\psi_c^t = \psi_{\Lambda-c+1}^t = 0,~\forall c\leq \min\left\{gt,\theta\right\}.
\end{align}
\end{proposition}
\begin{proof}
Appendix~\ref{pf_prop2}.
\end{proof}
Proposition~\ref{prop2} states that if $g=0$, the decoder fails to decode even a single column block of the message; otherwise, the entire message is decoded within  $\frac{\theta}{g}$ iterations. Here, it suffices to limit $g\leq \omega$ because $g\geq \omega$ means that the entire message is successfully decoded in the first iteration (Appendix~\ref{pf_implication_assumption}). 
Proposition~\ref{prop2} indicates that a SC-SPARC with UPA fails to decode at $R_U(P)\leq R<C(P)$ because the power \eqref{uniform_power} allocated to columns $1$ and $\Lambda$ of the base matrix is smaller than the power needed to make the event in $\psi_1^1$ \eqref{psi_SE} occur.

\section{V-power allocation}\label{Sec_alg}
%In Section~\ref{Sec_alg_1}, we present VPA and show that its output resembles the shape of the letter V. In Section~\ref{Sec_alg_2}, derive the PRF for VPA. In Section~\ref{Sec_alg_3}, we demonstrate the advantage of VPA over UPA by comparing their PRFs.

\subsection{VPA Algorithm}\label{Sec_alg_1}
Fixing an AWGN channel with noise variance $\sigma^2$ and a rate $R$, we present VPA for an $(\omega,\Lambda,P)$ base matrix.

In the extreme, a power allocation policy can allocate a different power to every non-zero entry of the base matrix $\mathsf W$. The output $\{\mathsf W_{rc}\}_{r\in[L_R], c\in[L_C]}$ of VPA satisfy:
\begin{itemize}
\item[a)] The power does not change with rows, i.e., $\forall c\in[L_C]$,
\begin{align}\label{col_unchange}
\mathsf W_{rc} \triangleq \mathsf W_c,~\forall c\leq r\leq c+\omega - 1;
\end{align}
\item[b)] The power is symmetric about the middle column index,
\begin{align} \label{symmetric}
&\mathsf W_{c} = \mathsf W_{\Lambda - c + 1}, \forall c\in[L_C].
\end{align}
\end{itemize}

We define function $\mathsf f_t\colon\mathbb R^{\theta -t + 1}\rightarrow \mathbb R$, $t= 1,\dots,\theta$  as\footnote{Although the summation in the denominator of the right side of \eqref{def_ft} may include $\mathsf W_{\theta+1},\dots, \mathsf W_{\Lambda - t + 1}$, $\mathsf f_t$ is still a function of variables $\mathsf W_t,\dots, \mathsf W_{\theta}$ only, due to the symmetry assumption \eqref{symmetric}.}
\begin{align}\label{def_ft}
\mathsf f_t\left(\{\mathsf W_i\}_{i=t}^{\theta}\right) 
\triangleq \sum_{r=t}^{t+\omega-1}\frac{\mathsf W_t}{\sigma^2 + \frac{1}{L_c}\sum_{c'=t}^{\min\{r,\Lambda - t + 1\}}\mathsf W_{c'}}.
\end{align}
Let $\{\delta_t\}_{t=1}^{\theta}$ be a sequence of positives chosen arbitrarily. 
\RestyleAlgo{ruled}
\begin{algorithm}
\caption{VPA}
\label{VPA}
\SetAlgoLined
\SetKwInOut{Input}{input}\SetKwInOut{Output}{output}
\Input{$\omega,\Lambda, R, P,\sigma, \{\delta_t\}_{t=1}^{\theta}$}
\Output{$\{\mathsf W_{c}\}_{c\in[L_C]}$}
\For {$t = \theta, \theta-1,\dots, 1$}{
Solve $\mathsf f_t\left(\{\mathsf W_i\}_{i=t}^{\theta}\right)  = 2RL_R$ for $\mathsf W_t$\;
$\mathsf W_t\leftarrow \mathsf W_t + \delta_t$\;
$\mathsf W_{\Lambda- t+1}\leftarrow \mathsf W_t$
}
\If{$\frac{1}{L_RL_c}\sum_{c=1}^{L_c}\omega \mathsf W_{c} > P$}{
Declare failure}
\If{$\frac{1}{L_RL_c}\sum_{c=1}^{L_c} \omega \mathsf W_{c} \leq P$}{
$\text{residual}\leftarrow P - \frac{1}{L_RL_c}\sum_{c=1}^{L_c}\omega \mathsf W_{c}$\;
$\mathsf W_1 \leftarrow \frac{\text{residual}L_RL_c}{2\omega}$\;
$\mathsf W_{\Lambda} \leftarrow \mathsf W_1$
}
\end{algorithm}

Proposition~\ref{prop3}, stated next, shows that VPA follows a shape of V, namely, $\mathsf W_c$ is non-increasing on $1\leq c\leq \theta$ and is non-decreasing on $\theta+1\leq c\leq \Lambda$ by symmetry \eqref{symmetric}. 
\begin{proposition}\label{prop3}
Power allocation $\mathsf W_1^{(V)},\mathsf W_2^{(V)},\dots,\mathsf W_{\theta}^{(V)}$ that ensure $\mathsf f_t = 2RL_R$ (line 2 of Algorithm~\ref{VPA}) for all $t=1,2,\dots, \theta$ are unique and satisfy 
\begin{align}\label{eq_match}
\mathsf W_1^{(V)}\geq \mathsf W_2^{(V)}\geq \dots \geq \mathsf W_\theta^{(V)}.
\end{align}
\end{proposition}
\begin{proof}
Appendix~\ref{pf_prop3}.
\end{proof}
% comment the following statement for the short version
Although the sequence $\{\mathsf W_t^{(V)}\}_{t=1}^\theta$ does not perfectly coincide with the sequence $\{\mathsf W_t\}_{t=1}^\theta$ formed at the end of line 5, it reflects the trend of $\{\mathsf W_t\}_{t=1}^\theta$ for abitrarily small $\{\delta_t\}_{t=1}^\theta$.

\subsection{VPA performance}\label{Sec_alg_2}
Before we show the PRF for VPA, we introduce Lemma~\ref{lemma} below. It states that if a column block of the message is decoded at some iteration, then it remains decoded in the subsequent iterations, and that the asymptotic SE $\psi_c^t$ \eqref{psi_SE} can be expressed in terms of $\mathsf f_t$ \eqref{def_ft} under some conditions.
\begin{lemma}\label{lemma}
Consider a SC-SPARC generated by an $(\omega,\Lambda, P)$ base matrix. Fix a noise variance $\sigma^2$ and a rate $R$.
\begin{enumerate}
\item If $\exists ~c\in[L_C], t\geq 1$, $\psi_c^t =0$, then $\psi_c^{s}=0$, $\forall s\geq t$.  \label{item1}
\item For a power allocation policy satisfying a)--b), at $t=1$, 
\begin{align}\label{psi_f11}
\psi_1^1 = 1 - \mathbbm{1}\left\{\mathsf f_1\left(\{\mathsf W_i\}_{i=1}^{\theta}\right)>2RL_R\right\};
\end{align}
if $\psi_{c}^{c} = 0, \forall c\leq t-1$, then at iterations $2\leq t\leq \theta$,
\begin{align}\label{psi_ftt}
\psi_t^t \leq 1 - \mathbbm{1}\left\{\mathsf f_t\left(\{\mathsf W_i\}_{i=t}^{\theta}\right) > 2RL_R\right\}.
\end{align} \label{item2}
\end{enumerate}
\end{lemma}
\begin{proof}
Appendix~\ref{pf_lemma}.
\end{proof} 
We present the PRF for VPA.
\begin{theorem}\label{thm_VPA}
Fix a finite coupling pair $(\omega,\Lambda)$ and an AWGN channel with noise variance $\sigma^2$. The PRF $P_V(R)$ \eqref{PRF} for VPA (Algorithm~\ref{VPA}) is given by
\begin{align}\label{PVRF}
P_V(R) =
\begin{cases}
\frac{2\omega}{L_RL_C}\sum_{c=1}^{\theta}\mathsf W_c^{(V)}, R< \frac{L_C(\omega+1)}{4L_R},~ \Lambda~\text{is even}\\
\frac{2\omega}{L_RL_C}\sum_{c=1}^{\theta-1}\mathsf W_c^{(V)}+ \mathsf W_\theta^{(V)}, R<\frac{L_C(\omega+2)}{4L_R}, \Lambda~\text{is odd}\\
\infty, ~\text{otherwise}.
\end{cases}
\end{align}
\end{theorem}
\begin{proof}[Proof sketch]
The proof is divided into two steps.

(i)  We show that VPA outputs $\{\mathsf W_c\}_{c\in[L_c]}$, or equivalently it does not declare failure, if and only if $P>P_V(R)$ and $R$ less than the upper bound in \eqref{PVRF}. Appendix~\ref{pf_thm2_step1}.

(ii) We show that the output $\{\mathsf W_c\}_{c\in[L_c]}$ of VPA ensures successful decoding. Appendix~\ref{pf_thm2_step2}. 
\end{proof}

The working principle of VPA is to allocate sufficient power to the outer columns of the base matrix in order to jumpstart the \emph{wave}-like decoding process that propagates from the sides to the middle of the message, and to allocate lower power (but not too low that prohibits the decoding process) to the inner columns of the base matrix.
%Lines~1--5 of VPA determines the power allocation. The power allocation $\{\mathsf W_c\}_{c\in[L_C]}$ determined by lines 1--5, though not necessarily satisfying the power constraint \eqref{power_constraint}, ensures successful decoding. This is a result of Lemmas~\ref{}

% If line 9 is satisfied, the residual power (line 10) is transferred to the first and the last column of the base matrix because (i) the transfer will not affect $\mathsf f_t>2RL_R$ for $t\geq 2$ as $\mathsf W_1$ is not a variable for $\mathsf f_t$, $t\geq 2$; (ii) the transfer will not affect $\mathsf f_1>2RL_R$ since $\mathsf f_1$ increaseas as $\mathsf W_1$ increases.

\subsection{VPA outperforms UPA}\label{Sec_alg_3}
 %We first observe that the upper bound on $R$ in \eqref{PURF} is smaller than or equal to that in \eqref{VR} since $\frac{L_C}{2L_R}\sum_{r=1}^{\omega}\frac{1}{r} \leq \frac{L_C}{2L_R}(1+\frac{\omega-1}{2})$, meaning that $P_V(R)$ is finite in a larger range of rates.
\begin{proposition}\label{prop_comp}
Fix a finite coupling pair $(\omega,\Lambda)$ and an AWGN channel with noise variance $\sigma^2$. The rate that ensures $P_V(R)<\infty$ also ensures $P_U(R)<\infty$, i.e.,
\begin{align}\label{rset}
\{R\colon P_U(R)<\infty\}\subseteq \{R\colon P_V(R)<\infty\}.
\end{align}
For a rate $R$ that belongs to both sets in \eqref{rset}, it holds that
\begin{align}\label{PVPU}
P_V(R)\leq P_U(R).
\end{align}
\end{proposition}
\begin{proof}
Appendix~\ref{pf_prop_comp}. 
\end{proof}
In fact, UPA is a special case of VPA by carefully selecting $\{\delta_t\}_{t=1}^{\theta}$ (Appendix~\ref{UPA_sc_VPA}).

\section{Simulations}\label{Sec_sim}

%We first show that $P_U(R)$ for UPA corresponds to the power that ensures $f_1 = 2RL_R$. Since $\mathsf f_t\leq \mathsf f_{t+1}$ for $\mathsf W_t = \mathsf W_{t+1}$, UPA leads to $\mathsf f_{\theta}\geq \dots\geq \mathsf f_1 = 2RL_R$. Since $\frac{\partial \mathsf f_t}{\partial \mathsf W_t}>0$ and $\frac{\partial \mathsf f_t}{\partial \mathsf W_s}<0$ for all $s\geq t+1$ and all $1\leq t\leq \theta$, there exists $\mathsf W_1,\dots, \mathsf W_{\theta}\leq P_U(R)\frac{L_R}{\omega}$ that satisfies $\mathsf f_t= 2RL_R$ for all $t=1,\dots,\theta$, and thus \eqref{} follows.

We use an example to illustrate \eqref{PVPU}. Consider $\omega=2$, $\Lambda=5$, $P=3$, $\sigma=1$, and $R =0.45$. For UPA, we have $\psi_1^1 = 1- \mathbbm{1}\{5.1708>5.4\} = 1$, and Proposition~\ref{prop2} implies that a SC-SPARC with UPA fails to decode the message. We now determine the power allocation using VPA. Choosing $\delta_t = 0.01$, $\forall t=1,2,3$, and following lines 1--5 of VPA, we obtain $\mathsf W_1 = 9.87, \mathsf W_2 = 8.74, \mathsf W_3 = 5.88$. We check that line 9 of VPA is satisfied, and we transfer the residual power to the boundary columns yielding $\mathsf W_1 = 10.82$. Since the power $P>P_V(R)$ in Theorem~\ref{thm_VPA}, the output of VPA ensures successful decoding.

%Proposition~\ref{prop3} also gives rise to a heuristic method to solve the equation in line 2. The equation in line 2 only contains one undetermined variable $W_t$, and solving the equation is equivalent to solving a degree-$\omega$ polynomial of $W_t$. Proposition~\ref{prop3} suggests that one can construct a domain set $\{\mathsf W_{t+1},\mathsf W_{t+1}+\Delta, \dots, \mathsf W_{t+1}+N\Delta\}$ for $\mathsf W_t$ with lower bound $\mathsf W_{t+1}$, resolution $\Delta$, and a large number $N$, and assign the smallest value in the set that satisfies $\mathsf f_t\geq 2RL_R$ to $\mathsf W_t$. 

%In fact, under assumptions a)--b), Algorithm~\ref{Alg1} is the \emph{optimal} algorithm in a sense that it yields power $\{W_c\}_{c\in[L_C]}$ satisfying a)--b) and ensuring the successful decoding of a SC-SPARCs as long as such power allocation exists.

While VPA is designed in the limit of large section length $M\rightarrow \infty$, we show by simulations that power allocation imitating the shape of the letter V \eqref{eq_match} also improves the finite-blocklength error performance of a SC-SPARC. We consider a SC-SPARC of parameters $M= 512,L=30, L_C = 15, L_R = 18, M_R = 12$ and an AWGN channel of variance $\sigma^2 = 1$. Fig.~\ref{BLER_comparison_PA} compare the SC-SPARC with UPA \eqref{uniform_power} and that with a VPA-like power allocation chosen empirically in Table~I. Fig.~\ref{BLER_comparison_PA} shows that the BLER of the VPA-like power allocation is smaller than that of UPA, especially in the middle part of the waterfall region. Fig.~\ref{BLER_iter} shows the convergence of the BLERs. To reduce the complexities, we use the Hadamard design matrix as in \cite{Barbier17}\hspace{1sp}\cite{Rush21}, instead of using the i.i.d. Gaussian design matrix. The simulations may not perfectly match our theoretical results since the asymptotic SE is accurate only for an i.i.d. Gaussian design matrix and $M\rightarrow\infty$.

\begin{figure}
\includegraphics[scale=0.55]{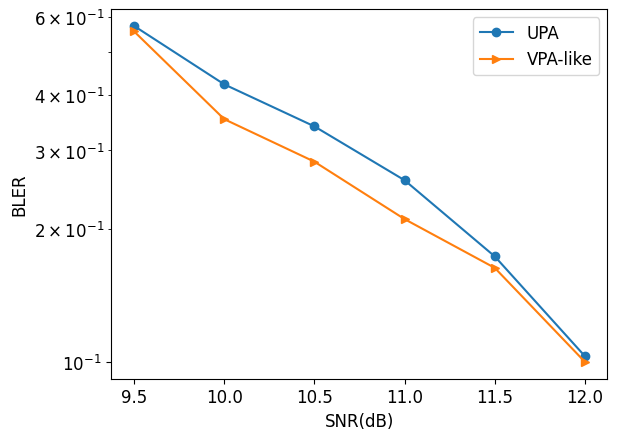}
\caption{Block error rate vs. SNR(dB).}
\label{BLER_comparison_PA}
\end{figure}

\begin{figure}
\includegraphics[scale=0.55]{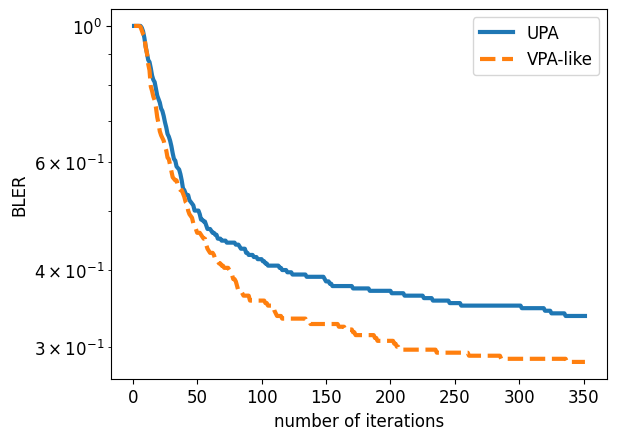}
\caption{Block error rate vs. number of iterations at SNR$=10.5$dB. }
\label{BLER_iter}
\end{figure}

\begin{table}[ht]
\centering
\caption{VPA-like power allocation \eqref{eq_match}}
\begin{tabular}{lcc} 
\toprule
SNR(dB) & Outer columns & Inner columns \\
\midrule
9.5 & $\mathsf W_1=\dots =\mathsf W_3 = 42.51$ &  $\mathsf W_4=\dots=\mathsf W_{8}=38.51$ \\ 
10.0 & $\mathsf W_1=\dots = \mathsf W_5 = 46.67$ & $\mathsf W_6=\dots=\mathsf W_{8}= 41.67$\\
10.5 & $\mathsf W_1=\dots=\mathsf W_4 = 52.36$ & $\mathsf W_5=\dots=\mathsf W_8= 48.36$\\
11.0 &  $\mathsf W_1=\dots=\mathsf W_5 =  58.32$ &  $\mathsf W_6=\dots=\mathsf W_{8}=53.32$\\
11.5 &   $\mathsf W_1=\dots=\mathsf W_6 = 64.56$ &   $\mathsf W_7=\dots=\mathsf W_{8}=59.56$\\
12.0 & $\mathsf W_1=\dots=\mathsf W_6 = 72.52$ & $\mathsf W_7=\dots=\mathsf W_{8}= 66.52$\\
\bottomrule
\end{tabular}
\end{table}

\section{Conclusion}
In this paper, we propose V-power allocation for the base matrix of a SC-SPARC with a finite coupling pair. It yields power allocation that descends from the outer columns to the inner columns of the base matrix, resembling the shape of the letter V. By analyzing the PRFs, we show that given a code rate, V-power allocation ensures successful decoding for a wider range of power compared to uniform power allocation. Numerical simulations indicate that power allocation following the shape of the letter V reduces the finite-blocklength block error rates of a SC-SPARC.
%Interesting future research directions include designing an efficient power allocation algorithm that makes use of the non-asymptotic SE of the SC-SPARCs and drops assumptions a)--b).

\appendix
\subsection{Proof of Theorem~\ref{thm_UPA}}\label{pf_thm_UPA}
\subsubsection{Proof of $R_U(P)$}\label{pf_thm_UPA_sec1}

Before we prove $R_U(P)$ in \eqref{RUPF}, we first show that a SC-SPARC with UPA succesfully decodes the entire message if and only if $\psi_1^1=0$. If $\psi_1^1=0$, then $g\geq 1$ in \eqref{prop2_psic1}, and Proposition~\ref{prop2} implies that the decoding is successful. To prove the reverse direction, we prove its equivalent, namely, if $\psi_1^1 = 1$, then a SC-SPARC with UPA cannot decode successfully. Since $\psi_c^1$ \eqref{psi_SE} is non-decreasing on $1\leq c\leq \theta$, we conclude that $\psi_1^1= 1$ implies $\psi_c^1 =1$ for all $c\in[L_C]$. Thus, $g=0$ in \eqref{prop2_psic1}, and Proposition~\ref{prop2} implies the decoding failure.

We proceed to prove $R_U(P)$ \eqref{RUPF}. We write it as
\begin{subequations}\label{pf_thm_UPA_1}
\begin{align}\label{pf_thm_UPA_1b}
R_U(P)
& = \sup\{R^*\colon \forall R<R^*, \psi_1^1(R, P,\text{UPA})=0\}\\\label{pf_thm_UPA_1c}
& = \sup \{R^*\colon \psi_1^1(R^*,P,\text{UPA})=0\}\\ \label{pf_thm_UPA_1d}
& = \left\{R^*\colon \sum_{r=1}^{\omega}\frac{P\frac{L_R}{\omega}}{\sigma^2 + \frac{r}{L_C}P\frac{L_R}{\omega}}=2R^*L_R\right\},
\end{align}
\end{subequations}
where \eqref{pf_thm_UPA_1b} holds as we have proved that a SC-SPARC with UPA decodes successfully if and only if $\psi_1^1=0$; \eqref{pf_thm_UPA_1c} holds since 
\begin{align}\label{psi11UPA}
\psi_1^1(R,P,\text{UPA}) = 1 - \mathbbm{1}\left\{\sum_{r=1}^\omega \frac{P\frac{L_R}{\omega}}{\sigma^2 + \frac{1}{L_C}rP\frac{L_R}{\omega}}>2RL_R\right\}
\end{align}
 is non-decreasing as $R$ increases; \eqref{pf_thm_UPA_1d} holds since \eqref{pf_thm_UPA_1c} is equivalent to the supremum of $R$ that makes the event in \eqref{psi11UPA} occur. Thus, \eqref{RUPF} follows.

\subsubsection{Proof of $P_U(R)$}\label{pf_thm_UPA_sec2}
Before we prove $P_U(R)$ \eqref{RUPF}, we calcuclate the derivative of the left side of the event in \eqref{psi11UPA} with respect to $P$ as
\begin{subequations}\label{devent}
\begin{align}
\frac{\partial \sum_{r=1}^\omega \frac{P\frac{L_R}{\omega}}{\sigma^2 + \frac{1}{L_C}rP\frac{L_R}{\omega}}}{\partial P}& = \sum_{r=1}^\omega \frac{\frac{L_R}{\omega}\sigma^2 }{\left(\sigma^2 + \frac{1}{L_C}rP\frac{L_R}{\omega}\right)^2}\\
&>0.
\end{align}
\end{subequations}
Since the left side of the event in \eqref{psi11UPA} increases as $P$ increases, we conclude that $\psi_1^1(R,P,\text{UPA})$ is non-increasing as $P$ increases.

To express $P_U(R)$ in terms of the inverse function $R_U^{-1}$ of $R_U(R)$, we show that the inverse function $R_U^{-1}$ exists. We calculate the derivative of $R_U(P)$ with respect to $P$ as
\begin{subequations}\label{dRUP}
\begin{align}
\frac{d R_U(P)}{d P} &= \frac{L_C}{2L_R}\sum_{r=1}^{\omega}\frac{\frac{L_C}{L_R}\frac{\sigma^2}{P^2}\omega}{(r+\frac{L_C}{L_R}\frac{\sigma^2}{P}\omega)^2}\\
&>0.
\end{align}
\end{subequations}
Since $R_U(P)$ is differentiable and its derivative is positive, we conclude that $R_U(P)$ is continuous and monotone, thus $R_U(P)$ is bijective and has an inverse function.

To demonstrate the domain and the range of the inverse function $R_U^{-1}$, we show the range of $R_U(P)$. Since $R_U(P)$ increases as $P$ increases by \eqref{dRUP}, it holds that for $\forall P<\infty$,
\begin{subequations}
\begin{align}
R_U(P)& < R_U(\infty)\\ \label{RUUB}
      &= \frac{L_C}{2L_R}\sum_{r=1}^\omega\frac{1}{r},
\end{align}
\end{subequations}
meaning that the inverse function satisfies
\begin{align}\label{RIrange}
R_U^{-1}(R)<\infty, \text{if and only if}~R< \frac{L_C}{2L_R}\sum_{r=1}^\omega\frac{1}{r}.
\end{align}

We proceed to show $P_U(R)$ \eqref{RUPF}. We write it as
\begin{subequations}\label{pf_thm_UPA_2}
\begin{align} \label{pf_thm_UPA_2b}
P_U(R)
& = \inf\{P^*\colon \forall P>P^*, \psi_1^1(R,P,\text{UPA}) = 0\}\\\label{pf_thm_UPA_2c}
& = \inf\{P^*\colon \psi_1^1(R,P^*,\text{UPA}) = 0\}\\\label{pf_thm_UPA_2d}
& = \left\{P^*\colon \sum_{r=1}^{\omega}\frac{P^*\frac{L_R}{\omega}}{\sigma^2 + \frac{r}{L_C}P^*\frac{L_R}{\omega}}=2RL_R \right\}\\\label{pf_thm_UPA_2e}
& = R^{-1}(R),
\end{align}
\end{subequations}
where \eqref{pf_thm_UPA_2b} holds as we have shown in Appendix~\ref{pf_thm_UPA_sec1} that a SC-SPARC with UPA decodes successfully if and only if $\psi_1^1 = 0$; \eqref{pf_thm_UPA_2c} holds as we have shown that $\psi_1^1(R,P,\text{UPA})$ is non-increasing as $P$ increases below \eqref{devent}; \eqref{pf_thm_UPA_2d} holds since \eqref{pf_thm_UPA_2c} is equivalent to the infimum of $P$ that makes the event in  \eqref{psi11UPA} occur; \eqref{pf_thm_UPA_2e} holds by noticing that the objective functions in \eqref{pf_thm_UPA_2d} and \eqref{pf_thm_UPA_1d} are the same and by the fact that $R_U^{-1}$ exists.

Plugging \eqref{RIrange} into \eqref{pf_thm_UPA_2e}, we obtain \eqref{PURF}.

\subsection{Proof of Proposition~\ref{prop2}}\label{pf_prop2}
We show \eqref{psi_gct} by mathematical induction. We denote by $\bar{\mathsf W}$ the non-zero value of a base matrix with UPA \eqref{uniform_power}. Plugging \eqref{phi_SE} into \eqref{psi_SE}, we write the asymptotic SE parameter $\psi_c^t$ as
\begin{align}\label{psi_ct}
\psi_c^t = 1 - \mathbbm{1}\left\{\sum_{r=c}^{c+\omega-1}\frac{\mathsf W_{rc}}{\sigma^2 + \frac{1}{L_c}\sum_{c'=\underline{c}_r}^{\bar c_r}\mathsf W_{rc'}\psi_{c'}^{t-1}}>2RL_R \right\}
\end{align}
where Definition~\ref{def_base_matrix} i)--ii) implies
\begin{align}\label{cr1}
&\underline{c}_r \triangleq \max\{1,r-\omega + 1\},\\ \label{cr2}
&\bar{c}_r \triangleq \min\{\Lambda, r\}.
\end{align}

\textbf{Initial step}: At iteration $t=1$, by the assumption of Proposition~\ref{prop2}, $\psi_c^1=0$, $\forall c\leq g$. Plugging $t\leftarrow 1$ into \eqref{psi_ct} and using $g\leq \omega$, we write $\psi_c^1$, $c\leq g$ as
\begin{subequations}\label{psi_c1}
\begin{align}\label{psi_c1a}
\psi_c^1 &= 1-\mathbbm{1}\left\{\sum_{r=c}^{c+\omega-1}\frac{\bar{\mathsf W}}{\sigma^2 + \frac{1}{L_c}\min\{\omega,r\}\bar{\mathsf W}}>2RL_R\right\}\\
&= 0.
\end{align}
\end{subequations}

\textbf{Induction step}: Assuming that \eqref{psi_gct} holds at iteration $t$, we show that it continues to hold at iteration $t+1$. If $t\geq \frac{\theta}{g}$, then Lemma~\ref{lemma} item \ref{item1}) implies that \eqref{psi_gct} holds at iteration $t+1$. If $t<\frac{\theta}{g}$, the asymptotic SE $\psi_c^{t+1}$ can be upper bounded as
\begin{subequations}
\begin{align}\nonumber
&\psi_c^{t+1} \\ \label{psi_a}
\leq~& 1 - \mathbbm{1}\left\{\sum_{r=c}^{c+\omega - 1} \frac{\bar{\mathsf W}}{\sigma^2 + \frac{1}{L_c}\sum_{c'=r-\omega+1}^{r}\bar{\mathsf W}\psi_{c'}^{t}}>2RL_R\right\}
\\\label{psi_b}
 =~& 1 - \mathbbm{1}\left\{\sum_{r=c}^{c+\omega - 1} \frac{\bar{\mathsf W}}{\sigma^2 + \frac{1}{L_c}\sum_{c'=\max\{r-\omega+1,gt+1\}}^{r}\bar{\mathsf W}}>2RL_R\right\}\\\label{psi_c}
=~& 1 - \mathbbm{1}\left\{\sum_{r=c}^{c+\omega - 1} \frac{\bar{\mathsf W}}{\sigma^2 + \frac{1}{L_c}\min\{\omega, r- gt\}\bar{\mathsf W}}>2RL_R\right\}\\ \label{psi_d}
=~& 1 -\mathbbm{1}\left\{\sum_{r=c-gt}^{c-gt+\omega-1}\frac{\bar{\mathsf W}}{\sigma^2 + \frac{1}{Lc}\min\{\omega,r\}\bar{\mathsf W}}>2RL_R\right\},
\end{align} 
where \eqref{psi_a} holds by plugging $t\leftarrow t+1$, $\bar c_r\leq r$, and $\underline{c}_r\geq r-\omega + 1$ into \eqref{psi_ct}; \eqref{psi_b} holds by the induction assumption and the fact $t< \frac{\theta}{g}$; \eqref{psi_c} holds by rewriting the summation in the denominator of \eqref{psi_b}; \eqref{psi_d} holds by change of measure $r\leftarrow r+gt$. Comparing \eqref{psi_c1} and \eqref{psi_d}, we conclude that
\begin{align}\label{pf_psi_c_t_1}
\psi_c^{t+1} = 0,~\forall gt\leq c\leq g(t+1).
\end{align}
\end{subequations}
Using \eqref{pf_psi_c_t_1}, the induction assumption, and Lemma~\ref{lemma} item \ref{item1}), we conclude that \eqref{psi_gct} holds at iteration $t+1$.

\subsection{$g\leq\omega$ is sufficient}\label{pf_implication_assumption}
We show that in Proposition~\ref{prop2}, it suffices to limit $g\leq \omega$ since $g\geq \omega$ implies that the entire message is successfully decoded in the first iteration. Indeed, for UPA \eqref{uniform_power}, $\psi_c^1$ \eqref{psi_c1a} is non-decreasing on $1\leq c\leq \omega$, remains constant on $\omega\leq c\leq \theta$, and is symmetric about $c=\theta$. As a result, $\psi_\omega^1 = 0$ implies $\psi_c^1 = 0$ for all $c\in[L_C]$.

\subsection{Proof of Proposition~\ref{prop3}}\label{pf_prop3}
In Appendix~\ref{Appen_D_main}, we first show that the sequence of power allocation $\left\{\mathsf W_i^{(V)}\right\}_{i=1}^\theta$ is unique, and we then show that it is non-increasing \eqref{eq_match}. In Appendices~\ref{pf_lemma_ftWt}--\ref{pf_lemma_D}, we prove the lemmas used in Appendix~\ref{Appen_D_main}.

\subsubsection{Main proof}\label{Appen_D_main}
To show the uniqueness, we introduce Lemma~\ref{lemma_ftWt} below.
\begin{lemma}\label{lemma_ftWt}
Fixing $\{\mathsf W_i\}_{i=t+1}^\theta$, function $\mathsf f_t$ is continuous in $\mathsf W_t$ and is monotonically increasing as $\mathsf W_t$ increases.
\end{lemma}
\begin{proof}
Appendix~\ref{pf_lemma_ftWt}.
\end{proof}
Lemma~\ref{lemma_ftWt} indicates that fixing $\{\mathsf W_i\}_{i=t+1}^\theta$, $\mathsf f_t$ is a bijective function of $\mathsf W_t$. Thus, there exist unique power allocation $\mathsf W_{\theta}^{(V)},\dots,\mathsf W_1^{(V)}$ that satisfy $\mathsf f_{\theta} =\dots = \mathsf f_1 = 2RL_R$.

We proceed to show that the sequence $\{\mathsf W_i^{(V)}\}_{i=1}^\theta$ is non-increasing~\eqref{eq_match}.
\begin{lemma}\label{lemma_D}
For any $t=1,2,\dots,\theta-1$, given $\mathsf W_{t+1}\geq \mathsf W_{t+2}\geq \dots\geq \mathsf W_\theta$, if $\mathsf W_t = \mathsf W_{t+1}$, it holds that $\mathsf f_t(\{\mathsf W_i\}_{i=t}^\theta)\leq \mathsf f_{t+1}(\{\mathsf W_i\}_{i=t+1}^\theta)$.
\end{lemma}
\begin{proof}
Appendix~\ref{pf_lemma_D}.
\end{proof}
The sequence $\{\mathsf W_i^{(V)}\}_{i=1}^\theta$ satisfies
\begin{align}\label{D1}
\mathsf f_t\left(\{\mathsf W_i^{(V)}\}_{i=t}^\theta\right) = \mathsf f_{t+1}\left(\{\mathsf W_i^{(V)}\}_{i=t+1}^\theta\right)
\end{align}
for all $t=1,2,\dots,\theta-1$. At $t=\theta-1$, Lemmas~\ref{lemma_ftWt}--\ref{lemma_D} and \eqref{D1} imply that $\mathsf W_{\theta-1}^{(V)}\geq \mathsf W_{\theta}^{(V)}$. Similarly, iteratively applying Lemmas~\ref{lemma_ftWt}--\ref{lemma_D} to $t=\theta-2,\theta-3,\dots,1$ in the backward manner, we conclude \eqref{eq_match}.

\subsubsection{Proof of Lemma~\ref{lemma_ftWt}}\label{pf_lemma_ftWt}
We compute the derivative of $\mathsf f_t$ with respect to $\mathsf W_t$. If $t+\omega-1<\Lambda-t + 1$, 
\begin{align}\label{derivative_ft1}
\frac{\partial \mathsf f_t}{\partial \mathsf W_t} = \sum_{r=t}^{t+\omega-1}\frac{\sigma^2 + \frac{1}{L_c}\sum_{c'=t+1}^r \mathsf W_{c'}}{(\sigma^2 + \frac{1}{L_c}\sum_{c'=t}^{r}\mathsf W_{c'})^2};
\end{align}
if $t+\omega-1\geq\Lambda-t + 1$,
\begin{align}\nonumber
\frac{\partial \mathsf f_t}{\partial \mathsf W_t}  &=  \sum_{r=t}^{\Lambda-t} \frac{\sigma^2 + \frac{1}{L_c}\sum_{c'=t+1}^r \mathsf W_{c'}}{(\sigma^2 + \frac{1}{L_c}\sum_{c'=t}^{r}\mathsf W_{c'})^2}\\\label{derivative_ft2}
& +(2t+\omega-\Lambda-1)\frac{\sigma^2 + \frac{1}{L_c}\sum_{c'=t+1}^{\Lambda-t} \mathsf W_{c'}}{(\sigma^2 + \frac{1}{L_c}\sum_{c'=t}^{\Lambda-t+1}\mathsf W_{c'})^2}
\end{align}
Since $\mathsf f_t$ is differentiable and its derivative is positive, we conclude that Lemma~\ref{lemma_ftWt} holds.

\subsubsection{Proof of Lemma~\ref{lemma_D}}\label{pf_lemma_D}

Given $\mathsf W_{t+1}\geq \dots\geq \mathsf W_\theta$, function $\mathsf f_{t+1}$ can be written as
\begin{subequations}
\begin{align}\nonumber
&\mathsf f_{t+1} \left(\{\mathsf W_i\}_{t+1}^\theta\right)\\\label{prop_1a}
 =~& \sum_{r=t+1}^{t+\omega} \frac{\mathsf W_{t+1}}{\sigma^2 + \frac{1}{L_c}\sum_{c'=t+1}^{\min\{r,\Lambda - t\}}\mathsf W_{c'}}\\ \label{prop_1b}
=~& \sum_{r=t}^{t+\omega-1} \frac{\mathsf W_{t+1}}{\sigma^2 + \frac{1}{L_c}\sum_{c'=t+1}^{\min\{r+1,\Lambda - t\}}\mathsf W_{c'}}\\ \label{prop_1c}
=~&  \sum_{r=t}^{t+\omega - 1} \frac{\mathsf W_{t+1}}{\sigma^2 + \frac{1}{L_c}\left(\mathsf W_{t+1} + \sum_{c'=t+2}^{\min\{r+1,\Lambda - t\}}\mathsf W_{c'}\right)}
\end{align}
\end{subequations}
where \eqref{prop_1a} holds by definition~\eqref{def_ft}; \eqref{prop_1b} holds by change of measure $r\leftarrow r+1$; \eqref{prop_1c} holds by expanding the summation in the denominator of \eqref{prop_1b}.
Function $\mathsf f_t$ with $\mathsf W_{t}\leftarrow \mathsf W_{t+1}$ can be written as
\begin{align} \nonumber
&\mathsf f_t \left(\mathsf W_{t+1},\{\mathsf W_{i}\}_{i=t+1}^{\theta}\right) \\ \label{prop_2b}
= &\sum_{r=t}^{t+\omega - 1} \frac{\mathsf W_{t+1}}{\sigma^2 + \frac{1}{L_c}\left(\mathsf W_{t+1}+\sum_{c'=t+1}^{\min\{r,\Lambda - t+1\}}\mathsf W_{c'}\right)}.
\end{align}

To compare \eqref{prop_1c} and \eqref{prop_2b}, it suffices to compare the summations in their denominators. We denote by $D_t$ and  $D_{t+1}$ the summations in the denominators of \eqref{prop_2b} and \eqref{prop_1c}, respectively, i.e.,
\begin{align}
D_{t} &\triangleq \sum_{c'=t+1}^{\min\{r,\Lambda - t+1\}}\mathsf W_{c'}\\
D_{t+1} &\triangleq \sum_{c'=t+2}^{\min\{r+1,\Lambda - t\}}\mathsf W_{c'}.
\end{align}
Fix $r=t,\dots,t+\omega - 1$.\\ 
\textbf{Case 1:} If $r\leq \Lambda - t + 1$ and $r+1\leq \Lambda - t$, it holds that
\begin{align}
D_{t} - D_{t+1} &= \mathsf W_{t+1} - \mathsf W_{r+1}\\ \label{Case1b}
 &\geq 0
\end{align}
where \eqref{Case1b} holds by the fact $t+1\leq r+1\leq \Lambda - t$ and the fact $\mathsf W_{t+1}\geq \dots\geq \mathsf W_{\theta}$.\\
\textbf{Case 2:} If $r\leq \Lambda - t + 1$ and $r+1> \Lambda - t$, it holds that
\begin{align}\label{Case2a}
D_t &= \sum_{c'=t+1}^{r}\mathsf W_{c'} \\ \label{Case2b}
&\geq \sum_{c'=t+1}^{\Lambda-t}\mathsf W_{c'}\\ \label{Case2c}
& \geq D_{t+1},
\end{align}
where \eqref{Case2b} holds since the assumptions on $r$ in Case 2 imply $r\in\{\Lambda-t, \Lambda-t + 1\}$.\\
\textbf{Case 3:} If $r>\Lambda -t + 1$ and $r+1>\Lambda-t$, it holds that
\begin{align}
D_t - D_{t+1} = \mathsf W_{t+1} + \mathsf W_{\Lambda-t+1} \geq 0.
\end{align} 

Since cases 1--3 indicate $D_t\geq D_{t+1}$, we conclude that if $\mathsf W_t =\mathsf W_{t+1}$, then $\mathsf f_t\leq \mathsf f_{t+1}$.

\subsection{Proof of Lemma~\ref{lemma}}\label{pf_lemma}
%We prove item~\ref{item1} in Appendix~\ref{pf_lemma_item1} and prove item~\ref{item2} in Appendix~\ref{pf_lemma_item2}.
\subsubsection{Proof of item~\ref{item1})}\label{pf_lemma_item1}
We prove item~\ref{item1}) by mathematical induction. We denote the set of zero positions of $\psi_c^t$ by
\begin{align}
\mathcal N^t \triangleq \{c\in[L_c]\colon \psi_c^t = 0\}.
\end{align}
To show item~\ref{item1}), it suffices to show
\begin{align}\label{N0}
\mathcal N^0\subseteq \mathcal N^1 \subseteq N^2 \subseteq \dots
\end{align}

\textbf{Initial step}: At $t=0$, $\psi_c^0=1$ for all $c\in[L_c]$, thus $\mathcal N^0$ is an empty set. It is trivial to conclude $\mathcal N^0\subseteq \mathcal N^1$.

\textbf{Induction step}: Assuming that $\mathcal N^{t-1}\subseteq \mathcal N^t$, we proceed to show $\mathcal N^t\subseteq \mathcal N^{t+1}$. The asymptotic SE \eqref{psi_SE} at iteration $t$ is given by \eqref{psi_ct}.
The induction assumption posits that if $\psi_c^{t-1}= 0$, we have $\psi_c^t=0$. As a result, the denominator of the event in \eqref{psi_ct} at iteration $t$ is larger than or equal to that at iteration $t+1$, and we obtain $\psi_c^t \geq \psi_c^{t+1}$. Since $\psi_c^t\in\{0,1\}$ is binary for all $c\in[L_C]$, $t=1,2,\dots$, we conclude \eqref{N0}.

\subsubsection{Proof of item~\ref{item2})}\label{pf_lemma_item2}
The asymptotic SE $\psi_1^1$ can be written as \eqref{psi_f11} by comparing \eqref{psi_f11} and \eqref{psi_ct} with $c\leftarrow 1$, $t\leftarrow 1$. It remains to show that \eqref{psi_ftt} holds for $t=2,\dots,\theta$ under the assumption that $\psi_c^c=0$ for $c\leq t-1$ in Lemma~\ref{lemma}. The SE parameter $\psi_t^t$ is given by \eqref{psi_ct} with $c\leftarrow t$ and the left side of its event can be lower bounded as
\begin{subequations}
\begin{align}\label{psi_ftt_a}
&\sum_{r=t}^{t+\omega - 1} \frac{\mathsf W_t}{\sigma^2 + \frac{1}{L_c}\sum_{c'=\underline c_r}^{\bar c_r}\mathsf W_{c'}\psi_{c'}^{t-1}} \\\label{psi_ftt_b}
\geq~&  \sum_{r=t}^{t+\omega - 1} \frac{\mathsf W_t}{\sigma^2 + \frac{1}{L_c}\sum_{c'=\max\{1,r-\omega + 1, t\}}^{\min\{r,\Lambda, \Lambda -t + 1\}}\mathsf W_{c'}}\\\label{psi_ftt_c}
=~& \sum_{r=t}^{t+\omega - 1} \frac{\mathsf W_t}{\sigma^2 + \frac{1}{L_c}\sum_{c'=t }^{\min\{r,\Lambda -t + 1\}}\mathsf W_{c'}}\\\label{psi_ftt_d}
=~&\mathsf f_t(\{\mathsf W_i\}_{i=t}^\theta),
\end{align}
\end{subequations}
where \eqref{psi_ftt_b} holds by \eqref{cr1}--\eqref{cr2}, the assumption $\psi_c^c=0$ for $c\leq t-1$, and the symmetry of $\psi_c^c$ \eqref{symmetric}; \eqref{psi_ftt_c} holds since $\Lambda -t + 1\leq \Lambda$, $t\geq 2$, and $r-\omega + 1 \leq t$; \eqref{psi_ftt_d} holds by definition \eqref{def_ft}. The equality of \eqref{psi_ftt_b} is achieved if and only if $\psi_c^{t-1}=1$ for all $t\leq c\leq \theta$. Replacing the left side of the event of $\psi_t^t$ by its lower bound in \eqref{psi_ftt_d}, we obtain \eqref{psi_ftt}.

\subsection{Proof of Theorem~\ref{thm_VPA}: step (i)}\label{pf_thm2_step1}
Given power $P<\infty$ and rate $R<\infty$, we show that VPA outputs $\{\mathsf W_{c}\}_{c\in[L_C]}$, i.e., it does not declare failure, if and only if $P>P_V(R)$ and $R$ is less than the upper bound in \eqref{PVRF}. To this end, we first introduce useful lemmas and notations in Appendix~\ref{Appen_l_and_n}; we prove the `if' direction in Appendix~\ref{Appen_if}; we prove the `only if' direction in Appendix~\ref{Appen_onlyif}; the proof of the lemmas in Appendix~\ref{Appen_l_and_n} are presented in Appendices~\ref{pf_lemma_R_exists}--\ref{pf_lemma_Mrt}.

\subsubsection{Lemmas and notations}\label{Appen_l_and_n}
We introduce Lemmas~\ref{lemma_R_exists}--\ref{lemma_Mrt}.
We denote by $\bar R$ the upper bound on $R$ in \eqref{PVRF}, i.e.,
\begin{align}
\bar R \triangleq \begin{cases}
\frac{L_C(\omega+1)}{4L_R}, & \Lambda~\text{is even},\\
\frac{L_C(\omega+2)}{4L_R}, & \Lambda ~\text{is odd}.
\end{cases}
\end{align}
Lemma~\ref{lemma_R_exists}, stated next, shows the existence of $\{\mathsf W_i^{(V)}\}_{i=1}^\theta$.
\begin{lemma}\label{lemma_R_exists}
If and only if $R<\bar R$, there exists a sequence $\mathsf W_1^{(V)},\mathsf W_2^{(V)},\dots,\mathsf W_\theta^{(V)}<\infty$ that satisfies 
$\mathsf f_t = 2RL_R$ (line 2 of VPA) simultaneously for all $t=1,2,\dots,\theta$.
\end{lemma}
\begin{proof}
Appendix~\ref{pf_lemma_R_exists}.
\end{proof}
Lemma~\ref{lemma_dftWt} shows how $\frac{\partial \mathsf f_t}{\partial \mathsf W_t}$ changes with $\mathsf W_t$.
\begin{lemma}\label{lemma_dftWt}
Fixing $\{\mathsf W_i\}_{i=t+1}^\theta$, the derivative $\frac{\partial \mathsf f_t}{\partial \mathsf W_t}$ \eqref{derivative_ft1}--\eqref{derivative_ft2} monotonically decreases as $\mathsf W_t$ increases.
\end{lemma}
\begin{proof}
Appendix~\ref{pf_lemma_dftWt}.
\end{proof}
Lemma~\ref{lemma_ftWs} below shows how $\mathsf f_t$ changes with $\mathsf W_s$, $s\geq t+1$.
\begin{lemma}\label{lemma_ftWs}
Fixing $\mathsf W_i$ for $t\leq i\leq \theta$, $i\neq s$, $s\geq t+1$, function $\mathsf f_t$ is continuous in $\mathsf W_s$ and is monotonically decreasing as $\mathsf W_s$ increases.
\end{lemma}
\begin{proof}
Appendix~\ref{pf_lemma_ftWs}.
\end{proof}

Lemma~\ref{lemma_Mrt} below shows how $\frac{\partial \mathsf f_t}{\partial \mathsf W_t}$ changes with $\{\mathsf W_i\}_{i=t+1}^\theta$.
\begin{lemma}\label{lemma_Mrt}
Consider $\mathsf W_i=\mathsf W_{\Lambda-i+1} \in[0, b_i]$, $t+1\leq i\leq\theta$. If the upper bounds of the intervals satisfy
\begin{align}\label{ub_geq_sqrt}
\sigma^2 + \frac{1}{L_C}\sum_{i=t+1}^{\Lambda-t}b_i\leq \sqrt{\frac{\mathsf W_t}{L_C}},
\end{align}
then the derivative $\frac{\partial \mathsf f_t}{\partial \mathsf W_t}$  \eqref{derivative_ft1}--\eqref{derivative_ft2} is non-decreasing as the elements in any non-empty subset of $\{\mathsf W_i\}_{i=t+1}^\theta$ increase on their corresponding intervals, $t=1,2,\dots,\theta$.
\end{lemma}
\begin{proof}
Appendix~\ref{pf_lemma_Mrt}.
\end{proof}

We introduce notations that will be used in the following proof. Fixing $\{\mathsf W_{i}^{(V)}\}_{i=t+1}^\theta$, we denote the derivative of $\mathsf f_t$ with respect to $\mathsf W_t$ at $\mathsf W_t = \bar {\mathsf W}_t$ by
\begin{align}
\mathsf f_t'(\bar {\mathsf W}_t)\triangleq \frac{\partial \mathsf f_t(\mathsf W_t, \{\mathsf W_{i}^{(V)}\}_{i=t+1}^\theta)}{\partial \mathsf W_t}\Big|_{\mathsf W_t = \bar {\mathsf W}_t}.
\end{align}
Given a sequence of positive numbers $\{\gamma_i\}_{i=t}^{\theta}$, we denote by $K_s^{(t)}$ a positive number that ensures
\begin{align}\nonumber
&\mathsf f_t\left(\{\mathsf W_i^{(V)}\}_{i=t}^{s-1},\{\mathsf W_i^{(V)}+\gamma_i\}_{i=s}^\theta\right) \\ \nonumber
-~& \mathsf f_t\left(\{\mathsf W_i^{(V)}\}_{i=t}^{s},\{\mathsf W_i^{(V)}+\gamma_i\}_{i=s+1}^\theta\right)  \\ \label{KST}
\geq~& -K_s^{(t)}\gamma_s,
\end{align}
for $s\geq t+1$.
Such $K_s^{(t)}$ always exists since $\mathsf f_t$ is continuously differentiable with respect to $\mathsf W_s$, i.e., it is a Lipschitz function of $\mathsf W_s$, and it decreases as $\mathsf W_s$ increases by Lemma~\ref{lemma_ftWs}.
Given $\{\bar {\mathsf W}_t\}_{t=1}^\theta$, we define sequence $\{g_t\}_{t=1}^{L_C}$,
\begin{subequations}
\begin{align}
&g_\theta\triangleq 1,\\ \label{gt_def}
&g_t \triangleq \frac{\sum_{s=t+1}^\theta K_s^{(t)}g_s}{\mathsf f_t'(\bar {\mathsf W}_t)},~t\leq \theta,\\
&g_t \triangleq g_{\Lambda-t+1}, ~t>\theta.
\end{align}
\end{subequations}
Given $\{\bar {\mathsf W}_t\}_{t=1}^{\theta-1}$ and an arbitrary positive number $\gamma_{\theta}^{(\theta)}$, we define a non-increasing sequence $\gamma_\theta^{(t)}$, $t=1,2,\dots,\theta-1$, as
\begin{align}\label{g_t_t}
\gamma_\theta^{(t)} \triangleq \min\left\{\gamma_\theta^{(t+1)},\frac{1}{g_t}(\bar {\mathsf W}_t - \mathsf W_t^{(V)})\right\}.
\end{align}
We denote
\begin{align}\label{g_t_m}
\gamma_{\theta}^{\max} \triangleq \max\left\{\gamma_\theta>0\colon \frac{\omega}{L_CL_R}\sum_{t=1}^{L_C}g_t\gamma_\theta = P - P_V(R)\right\}.
\end{align}
We define the minimum of \eqref{g_t_t} and \eqref{g_t_m} as
\begin{align}\label{g_th_s}
\gamma_{\theta}^* \triangleq \min\{\gamma_\theta^{(1)}, \gamma_\theta^{\max}\}. 
\end{align}
We define a sequence of numbers $\{\gamma_t^*\}_{t=1}^{\theta-1}$ as
\begin{align}\label{g_t_s}
\gamma_t^* \triangleq g_t\gamma_\theta^*.
\end{align}

\subsubsection{Proof of `If' direction}\label{Appen_if}
 We show that if $P>P_V(R)$ and $R<\bar R$, then VPA does not declare a failure, equivalently, there exists a sequence $\{\mathsf W_c\}_{c\in[L_C]}$ that satisfies (lines 1--5 and line 9 of Algorithm~\ref{VPA}):
\begin{align}\label{cons_1}
&\mathsf f_t(\{\mathsf W_i\}_{i=t}^{\theta})>2RL_R,\forall t=1,2,\dots,\theta,\\ \label{cons_2}
&\frac{\omega}{L_RL_C}\sum_{c=1}^{L_C}\mathsf W_c \leq P.
\end{align}

For $P>P_V(R)$ and $R<\bar R$, we set
\begin{align}\label{WPA}
\mathsf W_t = \mathsf W_t^{(V)} + \gamma_t^*,~t=1,2,\dots,\theta,
\end{align}
where $\mathsf W_t^{(V)}$ exists due to Lemma~\ref{lemma_R_exists}; $\gamma_t^*$ is defined in \eqref{g_t_s}; $\gamma_\theta^{(\theta)}$ defining $\gamma_\theta^*$ \eqref{g_th_s} via \eqref{g_t_t} is an arbitrary positive number; the sequence $\{\bar {\mathsf W}_t\}_{t=1}^{\theta-1}$ defining $\gamma_\theta^*$ \eqref{g_th_s} via \eqref{g_t_t} is chosen to be large enough so that
\begin{align}\label{WtWtV}
\bar {\mathsf W}_t > \mathsf W_t^{(V)},~\forall t=1,\dots,\theta-1,
\end{align} 
and that $\bar {\mathsf W}_t$ satisfies \eqref{ub_geq_sqrt} with $\mathsf W_t\leftarrow\bar {\mathsf W}_t$ and $b_i \leftarrow \mathsf W_i^{(V)} + g_i\gamma_\theta^{(t+1)}$, $i=t+1,\dots,\theta$.

We show that the power allocation in \eqref{WPA} satisfies \eqref{cons_1}--\eqref{cons_2}, respectively. The power allocation \eqref{WPA} satisfies the power constraint \eqref{cons_2} due to \eqref{g_t_m}--\eqref{g_th_s}. To show that the power allocation satisfies \eqref{cons_1}, it suffices to show that the following statement:
\begin{align}\label{statement}
\text{If}~0<\gamma_\theta^*\leq \gamma_\theta^{(t)}, ~\text{then \eqref{cons_1} holds at iteration}~ t.
\end{align}
 Since $0<\gamma_\theta^*\leq \gamma_{\theta}^{(t)}$ for all $t=1,2,\dots,\theta$ by \eqref{g_th_s}, this statement allows us to conclude that
 the condition \eqref{cons_1} is satisfied for all $t=1,2,\dots,\theta$.

It remains to prove the statement \eqref{statement}. The statement trivially holds for $t=\theta$, since $\mathsf W_\theta>\mathsf W_{\theta}^{(V)}$ ensures \eqref{cons_1} according to Lemma~\ref{lemma_ftWt}. We proceed to prove the statement for $t\leq \theta-1$.
Taking the difference between two $\mathsf f_t$ with different $\mathsf W_t$, we obtain
\begin{subequations}\label{Ft-ft}
\begin{align}\nonumber
&\mathsf f_t\left(\{\mathsf W_i^{(V)}+\gamma_i^*\}_{i=t}^\theta\right) - \mathsf f_t\left(\mathsf W_t^{(V)}, \{\mathsf W_i^{(V)}+\gamma_i^*\}_{i=t+1}^\theta\right)\\ \label{Ft-fta}
>~& \frac{\partial \mathsf f_t(\mathsf W_t, \{\mathsf W_{i}^{(V)}+\gamma_i^*\}_{i=t+1}^\theta)}{\partial \mathsf W_t}\Big|_{\mathsf W_t =\mathsf W_t^{(V)}+\gamma_t^*}\gamma_t^*\\\label{Ft-ftb}
\geq ~&  \frac{\partial \mathsf f_t(\mathsf W_t, \{\mathsf W_{i}^{(V)}+\gamma_i^*\}_{i=t+1}^\theta)}{\partial \mathsf W_t}\Big|_{\mathsf W_t =\bar{\mathsf W}_t} \gamma_t^*\\\label{Ft-ftc}
\geq ~& \mathsf f'_t(\bar {\mathsf W}_t) \gamma_t^*,
\end{align}
\end{subequations}
where \eqref{Ft-fta} holds by Mean Value Theorem and by Lemma~\ref{lemma_dftWt}; \eqref{Ft-ftb} holds due to $\gamma_t^* = g_t\gamma_\theta^*\leq g_t\gamma_\theta^{(t)}\leq \bar{\mathsf W}_t-\mathsf W_t^{(V)}$ and Lemma~\ref{lemma_dftWt}; \eqref{Ft-ftc} holds due to the fact that $\gamma_i^* = g_i\gamma_\theta^* \leq g_i\gamma_\theta^{(t)}\leq g_i\gamma_\theta^{(t+1)}$ for all $t+1\leq i \leq \theta$, the choice of $\bar {\mathsf W}_t$ below \eqref{WtWtV}, and Lemma~\ref{lemma_Mrt}.
We then take the difference between two $\mathsf f_t$ with different $\mathsf W_s$, $s\geq t+1$, just like \eqref{KST} with $\gamma_i\leftarrow \gamma_i^*$, $\forall i=s,\dots,\theta$. Summing \eqref{Ft-ft} and \eqref{KST} for all $s = t+1,t+2,\dots,\theta$, we obtain
\begin{subequations}\label{ftft_overall}
\begin{align}\label{ftft_overall_a}
&\mathsf f_t\left(\{\mathsf W_i^{(V)}+\gamma_i^*\}_{i=t}^\theta\right) - \mathsf f_t\left(\{\mathsf W_i^{(V)}\}_{i=t}^\theta\right) \\\label{ftft_overall_b}
>~& \mathsf f'_t(\bar {\mathsf W}_t) \gamma_t^* - \sum_{s=t+1}^\theta K_s^{(t)}\gamma_s^*\\\label{ftft_overall_c}
=~& \mathsf f'_t(\bar {\mathsf W}_t) g_t\gamma_\theta^* - \sum_{s=t+1}^\theta K_s^{(t)}g_s\gamma_\theta^*\\\label{ftft_overall_d}
= ~& 0,
\end{align}
\end{subequations}
where \eqref{ftft_overall_c} holds by plugging \eqref{g_t_s} into \eqref{ftft_overall_b}, and \eqref{ftft_overall_d} holds by plugging \eqref{gt_def} into \eqref{ftft_overall_c}. Since the second term in \eqref{ftft_overall_a} is equal to $2RL_R$, we conclude that \eqref{cons_1} holds at iteration $t$ with the power allocation in \eqref{WPA}.

\subsubsection{Proof of `Only if' direction}\label{Appen_onlyif}
Given $P<\infty$ and $R<\infty$, we show that if VPA does not declare failure, then $P>P_V(R)$ and $R<\bar R$. 

Not declaring failure implies that at the end of line 5 of Algorithm~\ref{VPA}, VPA forms finite $\{\mathsf W_t\}_{t=1}^\theta$ that ensure $\mathsf f_t>2RL_R$ for all $t=1,2,\dots,\theta$. We show by mathematical induction that there exist $\{\mathsf W_t^{(V)}\}_{t=1}^{\theta}$ that satisfy
\begin{subequations}\label{WtWtVlarge}
\begin{align}
&\mathsf f_t\left(\{\mathsf W_i^{(V)}\}_{i=t}^\theta\right) = 2RL_R\\\label{WtWtVlarge_b}
&\mathsf W_t^{(V)}<\mathsf W_t,t=1,2,\dots,\theta.
\end{align}
\end{subequations}
for all $1\leq t\leq \theta$ and for any $\{\mathsf W_t\}_{t=1}^\theta$ yielded by VPA.

\textbf{Initial step}: Since $\mathsf f_\theta(\mathsf W_\theta)>2RL_R$, $\mathsf f_\theta(0)=0$, and $\mathsf f_t$ is continuously increasing by Lemma~\ref{lemma_ftWt}, we conclude that there exists $\mathsf W_\theta^{(V)}$ that satisfies \eqref{WtWtVlarge} at $t=\theta$.

\textbf{Induction step}: Assuming that there exist $\{\mathsf W_t^{(V)}\}_{t=s+1}^{\theta}$ that satisfy \eqref{WtWtVlarge} for $t=s+1,\dots,\theta$, we show that together with $\{\mathsf W_t^{(V)}\}_{t=s+1}^{\theta}$, there exists $\mathsf W_s^{(V)}$ that satisfies \eqref{WtWtVlarge} at $t=s$. From Lemma~\ref{lemma_ftWs} and the induction assumption, we conclude 
\begin{align}
\mathsf f_s\left(\mathsf W_s,\{\mathsf W_t^{(V)}\}_{t=s+1}^{\theta}\right)>\mathsf f_s\left(\{\mathsf W_i\}_{i=s}^\theta\right).
\end{align}
Since $\mathsf f_s(\{\mathsf W_i\}_{i=s}^\theta)>2RL_R$, $\mathsf f_s(0,\{\mathsf W_t^{(V)}\}_{t=s+1}^{\theta})=0$, and $\mathsf f_s$ is continuously increasing in $\mathsf W_s$ by Lemma~\ref{lemma_ftWt}, there exists $\mathsf W_s^{(V)}$ that satisfies \eqref{WtWtVlarge}.

The existence of $\{\mathsf W_t^{(V)}\}_{t=1}^{\theta}$ satisfying \eqref{WtWtVlarge} implies $P>P_V(R)$ due to \eqref{WtWtVlarge_b} and implies $R<\bar R$ due to Lemma~\ref{lemma_R_exists}.

\subsubsection{Proof of Lemma~\ref{lemma_R_exists}}\label{pf_lemma_R_exists}
Fixing $\mathsf W_s=0$ for all $s\geq t+1$, we denote 
\begin{align}
R_t \triangleq \frac{1}{2L_R}\lim_{\mathsf W_t\rightarrow\infty}\mathsf f_t(\mathsf W_t,0,0,\dots,0).
\end{align}
Before we prove Lemma~\ref{lemma_R_exists}, we show 
\begin{align}\label{RRt}
\bar R = \min\{R_t, t=1,2,\dots,\theta\}.
\end{align}
 For $2t<\Lambda-\omega + 2$,
\begin{align}\label{Rt_ep1}
R_t = \frac{\omega L_C}{2L_R};
\end{align}
for $\Lambda-\omega+2<2t<\Lambda+1$, 
\begin{align}\label{Rt_ep2}
R_t = \frac{1}{2L_R}\left((\Lambda-2t+1)L_C + (2t+\omega-\Lambda-1)\frac{L_C}{2}\right);
\end{align}
for $2t = \Lambda+1$,
\begin{align}\label{Rt_ep3}
R_t = \frac{\omega L_C}{2L_R}.
\end{align}
If $\Lambda$ is even, $2\theta<\Lambda+1$, the first two cases \eqref{Rt_ep1}--\eqref{Rt_ep2} describe $R_t$ for all $t=1,2,\dots,\theta$, and the minimum in \eqref{RRt} is achieved at $2t = \Lambda$, yielding
\begin{align}
\bar R = R_{\frac{\Lambda}{2}}= \frac{L_C(\omega+1)}{4L_R}.
\end{align}
If $\Lambda$ is odd, $2\theta = \Lambda+1$, the three cases \eqref{Rt_ep1}--\eqref{Rt_ep3} jointly describe $R_t$ for all $t=1,2,\dots,\theta$, and the minimum in \eqref{RRt} is achieved at $2t = \Lambda-1$, yielding
\begin{align}
\bar R = R_{\frac{\Lambda-1}{2}} = \frac{L_C(\omega+2)}{4L_R}.
\end{align}

We begin to prove Lemma~\ref{lemma_R_exists}.

We show that if $R<\bar R$, there exist $\mathsf W_1^{(V)},\dots,\mathsf W_\theta^{(V)}<\infty$ that satisfy $\mathsf f_t =2RL_R$ for all $t=1,2,\dots,\theta$. We prove this by mathematical induction. 

\textbf{Initial step}: Since $R<\bar R \leq R_{\theta}$ and $\mathsf f_{\theta}\in[0,2R_{\theta}L_R)$ is continuously increasing in $\mathsf W_{\theta}$, there exists $\mathsf W_\theta^{(V)}<\infty$ that satisfies $\mathsf f_\theta(\mathsf W_\theta^{(V)})=2RL_R$. 

\textbf{Induction step}: Assuming there exist $\mathsf W_{t+1}^{(V)},\dots,\mathsf W_\theta^{(V)}<\infty$ that satisfy $\mathsf f_i=2RL_R$ for all $i=t+1,\dots,\theta$, we show that together with $\{\mathsf W_i^{(V)}\}_{i=t+1}^\theta$, there exists $\mathsf W_t^{(V)}<\infty$ that satisfies $\mathsf f_t=2RL_R$. Since $\{\mathsf W_i^{(V)}\}_{i=t+1}^\theta$ are finite by the induction assumption, it holds that
\begin{align}
\lim_{\mathsf W_t\rightarrow\infty}\mathsf f_t\left(\mathsf W_t, \{\mathsf W_i^{(V)}\}_{i=t+1}^\theta\right)= 2R_tL_R
\end{align}
Since $R<\bar R\leq R_t$ and $\mathsf f_{t}\in[0,2R_tL_R)$ is continuously increasing in $\mathsf W_t\in[0,\infty)$, there exists $\mathsf W_t<\infty$ that achieves $\mathsf f_t=2RL_R$.

We show that if there exist $\mathsf W_1^{(V)},\mathsf W_2^{(V)},\dots,\mathsf W_\theta^{(V)}<\infty$ that satisfy $\mathsf f_t=2RL_R$ for all $t=1,2,\dots,\theta$, then the rate satisfies $R<\bar R$. It holds that
\begin{subequations}\label{RLRft}
\begin{align}\label{RLRft_a}
2RL_R &= \mathsf f_t\left(\mathsf W_t^{(V)},\mathsf W_{t+1}^{(V)},\dots,\mathsf W_\theta^{(V)}\right)\\\label{RLRft_b}
&<\lim_{\mathsf W_t\rightarrow\infty}\mathsf f_t\left(\mathsf W_t,\mathsf W_{t+1}^{(V)},\dots,\mathsf W_\theta^{(V)}\right)\\\label{RLRft_c}
& = 2R_tL_R,
\end{align}
\end{subequations}
where \eqref{RLRft_b} is by Lemma~\ref{lemma_ftWt}; \eqref{RLRft_c} holds since $\{\mathsf W_i^{(V)}\}_{i=t+1}^\theta$ are finite. Since \eqref{RLRft} holds for all $t=1,2,\dots,\theta$, we conclude $R<\bar R$, where $\bar R$ is defined in \eqref{RRt}.

\subsubsection{Proof of Lemma~\ref{lemma_dftWt}}\label{pf_lemma_dftWt}
Since $\mathsf W_t$ only appears in the denominator of $\frac{\partial \mathsf f_t}{\partial W_t}$ \eqref{derivative_ft1}--\eqref{derivative_ft2} as a summand, the increase of $\mathsf W_t$ leads to the decrease of $\frac{\partial \mathsf f_t}{\partial W_t}$.

\subsubsection{Proof of Lemma~\ref{lemma_ftWs}}\label{pf_lemma_ftWs}
We show that the derivative of $\mathsf f_t$ with respect to $\mathsf W_s$ for $s\geq t+1$ is negative. For $\Lambda-s+1>\min\{r,\Lambda-t+1\}$,
\begin{align}
\frac{\partial \mathsf f_t}{\partial \mathsf W_s} = \sum_{r=t}^{t+\omega-1}\frac{-\mathsf W_t\frac{1}{L_C}}{\left(\sigma^2 + \frac{1}{L_C}\sum_{c'=t}^{\min\{r,\Lambda-t+1\}}\mathsf W_{c'}\right)^2}<0.
\end{align}
For $\Lambda-s+1\leq \min\{r,\Lambda-t+1\}$,
\begin{align}
\frac{\partial \mathsf f_t}{\partial \mathsf W_s} = \sum_{r=t}^{t+\omega-1}\frac{-\mathsf W_t\frac{2}{L_C}}{\left(\sigma^2 + \frac{1}{L_C}\sum_{c'=t}^{\min\{r,\Lambda-t+1\}}\mathsf W_{c'}\right)^2}<0.
\end{align}

\subsubsection{Proof of Lemma~\ref{lemma_Mrt}}\label{pf_lemma_Mrt}
We denote by  $M_{r,t+1} \triangleq \sigma^2 + \frac{1}{L_C}\sum_{i=t+1}^{r}\mathsf W_{i}$, $t+1\leq r\leq \Lambda-t$, and we rewrite $\frac{\partial\mathsf f_t}{\partial \mathsf W_t}$ in \eqref{derivative_ft1}--\eqref{derivative_ft2} as follows. If $t+\omega-1< \Lambda-t+1$,
\begin{subequations}\label{ftM1}
\begin{align}\label{ftM1a}
\frac{\partial \mathsf f_t}{\partial \mathsf W_t} &=\frac{\sigma^2}{(\sigma^2 + \frac{1}{L_C}\mathsf W_t)^2}\\  \label{ftM1b}
 &+ \sum_{r=t+1}^{t+\omega-1}\frac{1}{(\sqrt{M_{r,t+1}}+\frac{1}{L_C}\frac{\mathsf W_t}{\sqrt{M_{r,t+1}}})^2};
\end{align}
\end{subequations}
if $t+\omega-1\geq \Lambda-t+1$,
\begin{subequations}\label{ftM2}
\begin{align}\label{ftM20}
&\frac{\partial \mathsf f_t}{\partial \mathsf W_t}= \frac{\sigma^2}{(\sigma^2 + \frac{1}{L_C}\mathsf W_t)^2} \\\label{ftM2_a}
&+ \sum_{r=t+1}^{\Lambda-t}\frac{1}{(\sqrt{M_{r,t+1}}+\frac{1}{L_C}\frac{\mathsf W_t}{\sqrt{M_{r,t+1}}})^2}\\\label{ftM2_b}
&+ (2t+\omega-\Lambda-1)\frac{1}{(\sqrt{M_{\Lambda-t,t}}+\frac{2}{L_C}\frac{\mathsf W_t}{\sqrt{M_{\Lambda-t,t+1}}})^2}.
\end{align} 
\end{subequations}
We observe that (i) each summand in \eqref{ftM1b} and \eqref{ftM2_a} monotonically increases as $M_{r,t+1}$ increases on $M_{r,t+1}\in\left[0,\sqrt{\frac{\mathsf W_t}{L_C}}\right]$; (ii) \eqref{ftM2_b} increases as $M_{\Lambda-t,t+1}$ increases on $M_{\Lambda-t,t+1}\in\left[0,\sqrt{\frac{2\mathsf W_t}{L_C}}\right]$; (iii) $M_{r,t+1}$ increases as $r$ increases.

Since \eqref{ub_geq_sqrt} means $M_{\Lambda-t, t+1} \leq \sqrt{\frac{\mathsf W_t}{L_C}}$, observation (iii) implies that $\{\mathsf W_i\}_{i=t+1}^{\theta}$ satisfy $M_{r,t+1}\leq \sqrt{\frac{\mathsf W_t}{L_C}}$ for all $t+1\leq r\leq \Lambda-t$. Thus, observations (i)--(ii) imply that $\frac{\partial \mathsf f_t}{\partial \mathsf W_t}$ \eqref{ftM1}--\eqref{ftM2} is non-decreasing as the elements in any non-empty subset of $\{\mathsf W_i\}_{i=t+1}^{\theta}$ increase on their corresponding intervals.

\subsection{Proof of Theorem~\ref{thm_VPA}: step (ii)} \label{pf_thm2_step2} 
We show that the output power allocation of VPA ensures successful decoding.
The power determined at the end of line 5 of Algorithm~\ref{VPA} satisfies 
\begin{align}\label{ft_all}
\mathsf f_t\left(\{\mathsf W_i\}_{i=t}^{\theta}\right)  > 2RL_R,\forall t=1,\dots,\theta,
\end{align}
since $\delta_t>0$ for all $t=1,2,\dots,\theta$ and Lemma~\ref{lemma_ftWt}, which states that $\mathsf f_t$ increases as $\mathsf W_t$ increases. 
Plugging \eqref{ft_all} into Lemma~\ref{lemma} item~\ref{item2}), we conclude $\psi_c^{\theta}=0$, $\forall c\in[L_C]$, meaning that VPA ensures successful decoding within $\theta$ iterations. 

Since the left sides of the inequalites in lines 6 and 9 are equal to the left side of \eqref{power_constraint}, representing the resultant power, lines 6 and 9 check the satisfaction of the power constraint \eqref{power_constraint}. After transferring the resiudal power to $\mathsf W_1$ and $\mathsf W_{\Lambda}$ in lines 9--13, the resultant power still satisfies \eqref{ft_all} since $\mathsf f_{t}, t\geq 2$ does not depend on $\mathsf W_1$ and $\mathsf W_{\Lambda}$, and $\mathsf f_1$ by Lemma~\ref{lemma_ftWt} monotonically increases as $\mathsf W_1$ increases.

\subsection{Proof of Proposition~\ref{prop_comp}}\label{pf_prop_comp}
We first show \eqref{rset}. It suffices to show that the upper bound on $R$ in \eqref{PURF} is smaller than or equal to that in \eqref{PVRF}. For clarity, we denote the upper bounds on $R$ in \eqref{PURF} and \eqref{PVRF} by $\bar R_U$ and $\bar R_V$, respectively. We upper bound $\bar R_U$ as
\begin{subequations}\label{barRU}
\begin{align}\label{barRUa}
\bar R_U &= \frac{L_C}{2L_R}\sum_{r=1}^{\omega}\frac{1}{r}\\\label{barRUb}
& \leq \frac{L_C}{2L_R}\left(1+\frac{\omega-1}{2}\right)\\\label{barRUc}
& = \frac{L_C(\omega+1)}{4L_R}\\
&\leq \bar R_V
\end{align}
\end{subequations}
where \eqref{barRUb} holds by lower bounding $r$ by $2$ for all $r\geq 2$. %The equality holds if and only if $\omega \in\{1,2\}$.

Given rate $R$ that ensures $P_U(R)<\infty$ and $P_V(R)<\infty$, we proceed to show \eqref{PVPU}. We denote by $\bar {\mathsf W}  \triangleq P_U(R)\frac{L_R}{\omega}$ the UPA at $P_U(R)$. To show \eqref{PVPU}, it suffices to show 
\begin{align}\label{WtVW}
\mathsf W_t^{(V)}\leq \bar {\mathsf W}.
\end{align}
To this end, from \eqref{psi_f11} and \eqref{pf_thm_UPA_2}, we conclude
\begin{align}\label{f12RLR}
\mathsf f_1(\bar {\mathsf W},\bar {\mathsf W},\dots,\bar {\mathsf W}) = 2RL_R.
\end{align}
Lemma~\ref{lemma_D} implies for all $t=2,3,\dots,\theta$,
\begin{align}\label{ftgeqR}
\mathsf f_t(\bar {\mathsf W},\bar {\mathsf W},\dots,\bar {\mathsf W})\geq 2RL_R.
\end{align}
At $t=\theta$, Lemma~\ref{lemma_ftWt} and \eqref{ftgeqR} imply \eqref{WtVW}.
At $t=\theta-1$, since Lemma~\ref{lemma_ftWs} implies $\mathsf f_{\theta-1}(\bar{\mathsf W},\mathsf W_\theta^{(V)}) \geq \mathsf f_{\theta-1}(\bar{\mathsf W},\bar{\mathsf W})$,
we conclude from Lemma~\ref{lemma_ftWt} and \eqref{ftgeqR} that \eqref{WtVW} holds at $t=\theta-1$.
Simiarly, at $t=\theta-2,\theta-3,\dots,1$, iteratively using Lemmas~\ref{lemma_ftWt}~and~\ref{lemma_ftWs} and \eqref{f12RLR}--\eqref{ftgeqR}, we obtain \eqref{WtVW}. 

\subsection{UPA is a special case of VPA}\label{UPA_sc_VPA}

We show that UPA is a special case of VPA. This is an alternative proof for \eqref{PVPU}. Consider any $P>P_U(R)$ and $R<\bar R_U$ with UPA $\bar{\mathsf W}' \triangleq P\frac{L_R}{\omega}$. Due to the fact that $\mathsf f_1$ in \eqref{f12RLR} increases as $\bar {\mathsf W}$ increases, $\bar {\mathsf W}'>\bar {\mathsf W}$, and Lemma~\ref{lemma_D}, we conclude $\mathsf f_t(\bar{\mathsf W}',\bar{\mathsf W}',\dots,\bar{\mathsf W}')>2RL_R$ for all $t=1,2,\dots,\theta$. VPA recovers UPA by choosing $\delta_t = \bar{\mathsf W}' - \mathsf W_t$, where $\mathsf W_t$ is the output of line 2 of Algorithm~\ref{VPA}. The difference $\delta_t$ is positive for all $t=1,2,\dots,\theta$ since Lemma~\ref{lemma_ftWt} implies that the output of line 2 satisfies $\mathsf W_t<\bar{\mathsf W}'$.


\begin{thebibliography}{00}
\bibitem{Joseph12} A. Joseph and A. R. Barron, ``Least squares superposition codes of moderate dictionary size are reliable at rates up to capacity,'' in \textit{IEEE Trans. Inf. Theory}, vol. 58, no. 5, pp. 2541--2557, May 2012.
\bibitem{Barbier14} J. Barbier and F. Krzakala. ``Replica analysis and approximate message passing decoder for superposition codes,'' in \textit{2014 IEEE Int. Symp. Inf. Theory}, Honolulu, HI, USA, July 2014, pp. 1494-1498.
\bibitem{Rush17}C. Rush, A. Greig and R. Venkataramanan, ``Capacity-achieving sparse superposition codes via approximate message passing decoding,'' in \textit{IEEE Trans. Inf. Theory}, vol. 63, no. 3, pp. 1476--1500, March 2017.
\bibitem{Barbier17}J. Barbier and F. Krzakala, ``Approximate message-passing decoder and capacity achieving sparse superposition codes,'' in \textit{IEEE Trans. Inf. Theory}, vol. 63, no. 8, pp. 4894--4927, Aug. 2017.
\bibitem{Joseph14}A. Joseph and A. R. Barron, ``Fast sparse superposition codes have near exponential error probability for  $R<C$ ,'' in \textit{IEEE Trans. Inf. Theory}, vol. 60, no. 2, pp. 919--942, Feb. 2014. 
\bibitem{Cho13} S. Cho, and A. Barron. ``Approximate iterative Bayes optimal estimates for high-rate sparse superposition codes,'' in \textit{Sixth Workshop on Information-Theoretic Methods in Science and Engineering}, 2013.
\bibitem{Ramji19}R. Venkataramanan, S. Tatikonda, and A. Barron, ``Sparse regression
codes,'' in \textit{Foundations and Trends in Communications and Information Theory}, vol. 15, nos. 1–2,
pp. 1--195, 2019.
\bibitem{Barbier16}J. Barbier, M. Dia and N. Macris, ``Proof of threshold saturation for spatially coupled sparse superposition codes,'' in \textit{2016 IEEE Int. Symp. Inf. Theory}, Barcelona, Spain, 2016, pp. 1173--1177.
\bibitem{Rush21}C. Rush, K. Hsieh and R. Venkataramanan, ``Capacity-achieving spatially coupled sparse superposition codes with AMP decoding,'' in \textit{IEEE Trans. Inf. Theory}, vol. 67, no. 7, pp. 4446--4484, July 2021.
\bibitem{Hsieh18}K. Hsieh, C. Rush and R. Venkataramanan, ``Spatially coupled sparse regression codes: design and state evolution analysis,'' in \textit{2018 IEEE Int. Symp. Inf. Theory}, Vail, CO, USA, 2018, pp. 1016--1020.
\bibitem{Kudekar}S. Kudekar, T. J. Richardson and R. L. Urbanke, ``Threshold saturation via spatial coupling: why convolutional LDPC ensembles perform so well over the BEC,'' in \textit{IEEE Trans. Inf. Theory}, vol. 57, no. 2, pp. 803--834, Feb. 2011.
\bibitem{Mitchell}D. G. M. Mitchell, M. Lentmaier and D. J. Costello, ``Spatially Coupled LDPC Codes Constructed From Protographs,'' in \textit{IEEE Trans. Inf. Theory}, vol. 61, no. 9, pp. 4866--4889, Sept. 2015.
\bibitem{Rush19}C. Rush, K. Hsieh and R. Venkataramanan, ``Spatially coupled sparse regression codes with sliding window AMP decoding,'' in \textit{2019 IEEE Inf. Theory Workshop}, Visby, Sweden, 2019, pp. 1--5.
\bibitem{NG} N. Guo, S. Liang, W. Han, ``Power allocation for the base matrix of spatially coupled sparse regression codes'', \textit{Arxiv Preprint}, May 2023.
\end{thebibliography}
\end{document}